\pdfoutput=1

\documentclass[journal,twoside,web]{ieeecolor}

\usepackage{scalerel}
\usepackage{amsthm}
\usepackage{amsmath, amssymb, amsfonts}
\usepackage{generic}
\usepackage{float}
\usepackage[scr=rsfs]{mathalfa}
\usepackage{dashrule}
\usepackage{graphicx}
\usepackage{textcomp}
\usepackage{footnote}

\def\BibTeX{{\rm B\kern-.05em{\sc i\kern-.025em b}\kern-.08em
		T\kern-.1667em\lower.7ex\hbox{E}\kern-.125emX}}
\usepackage{cite}
\makeatletter
\let\NAT@parse\undefined
\makeatother
\usepackage[bookmarks]{hyperref}
\hypersetup{colorlinks = true, citecolor = red, linkcolor = blue, urlcolor = black}

\newtheorem{theorem}{Theorem} 
\newtheorem{remark}[theorem]{Remark} 
\newtheorem{lemma}[theorem]{Lemma}
\newtheorem{corollary}[theorem]{Corollary}
\newtheorem{proposition}[theorem]{Proposition}
\newtheorem{definition}[theorem]{Definition}

\numberwithin{theorem}{section}

\DeclareMathOperator*\uplim{\overline{lim}}

\def\hmath$#1${\texorpdfstring{{\fontsize{11}{11}\rmfamily\textit{#1}}}{#1}}

\newcommand{\neweq}[1]{\mathrel{\stackrel{\makebox[0pt]{\mbox{\normalfont\tiny{#1}}}}{=}}}
\newcommand{\newgeq}[1]{\mathrel{\stackrel{\makebox[0pt]{\mbox{\normalfont\tiny{#1}}}}{\geq}}}
\newcommand{\newleq}[1]{\mathrel{\stackrel{\makebox[0pt]{\mbox{\normalfont\tiny{#1}}}}{\leq}}}

\begin{document}
	
	\title{\LARGE An Information-Theoretic Analysis of Discrete-Time Control and Filtering Limitations by the I-MMSE Relationships}
	
	\author{Neng Wan$^*$, Dapeng Li$^*$, Naira Hovakimyan, \IEEEmembership{Fellow, IEEE}, and Petros G. Voulgaris, \IEEEmembership{Fellow, IEEE}
		\thanks{\hspace{-5pt}* Authors contributed equally to this paper.}
		\thanks{\hspace{-5pt}** The extended version of this paper with supplementary material and appendices can be found in \cite{Wan_arXiv_2023} and \cite[Chapter 3]{Wan_2024}.}
	       \thanks{\hspace{-5pt}*** This manuscript is the extended version of the paper with the same title accepted by IEEE Transactions on Automatic Control.}
		\thanks{\hspace{-5pt}**** This work was partially supported by AFOSR, NASA and NSF. }
		\thanks{N. Wan and N. Hovakimyan are with the Department of Mechanical Science and Engineering and Coordinated Science Laboratory, University of Illinois at Urbana-Champaign, Urbana, IL 61801, USA (e-mail: \{nengwan2, nhovakim\}@illinois.edu).}
		\thanks{D. Li is with the Department of Mechanical and Energy Engineering, Southern University of Science and Technology, Shenzhen, Guangdong 518055, China (e-mail: dapeng.ustc@gmail.com).}
		\thanks{P. G. Voulgaris is with the Department of Mechanical Engineering, University of Nevada, Reno, NV 89557, USA (e-mail: pvoulgaris@unr.edu).}}
	
	\maketitle

	\hspace{-6pt}
	\begin{abstract}
		Fundamental limitations or performance trade-offs/limits are important properties and constraints of both control and filtering systems. Among various trade-off metrics, total information rate that characterizes the sensitivity trade-offs and time-averaged performance of control and filtering systems was conventionally studied by using the differential entropy rate and Kolmogorov-Bode formula. In this paper, by extending the famous I-MMSE (mutual information -- minimum mean-square error) relationships to the discrete-time additive white Gaussian channels with and without feedback, a new paradigm is introduced to estimate and analyze total information rate as a control and filtering trade-off metric. Under this framework, we explore the trade-off properties of total information rate for a variety of the discrete-time control and filtering systems, e.g., LTI, LTV, and nonlinear, and propose an alternative approach to investigate total information rate via optimal estimation.
	\end{abstract}
	

	\vspace{-0.8em}

	\section{Introduction}\label{sec1}
	Fundamental limitations or performance trade-offs of control and filtering systems have been an important and long-lasting topic, since the introduction of Bode's integral in the 1940's. Some established trade-off metrics include but are not limited to the Bode's integral \cite{Bode_1945, Freudenberg_1985, Sung_IJC_1988, Stein_CSM_2003, Wan_TAC_2020, Wan_TAC_2023}, minimum cost function \cite{Seron_TAC_1999, Liu_CIS_2014, Kostina_TAC_2019}, data transmission rate \cite{Nair_PIEEE_2007}, total or directed information rate \cite{Elia_TAC_2004, Martins_TAC_2008, Li_TAC_2013, Charalambous_TAC_2017, Fang_2017}, and lowest achievable estimation error \cite{Braslavsky_Auto_1999, Wan_arxiv_2022}, all of which play an essential role in characterizing the fundamental conflicts or trade-offs between the limitations of physical systems and the pursuit of control (or estimation) performance. Complex analysis and information theory are two major tools for modeling and analyzing these control and filtering trade-offs. Complex analysis is convenient for studying the linear time-invariant (LTI) systems and capturing the frequency-domain trade-offs, such as Bode's integral and its variants. By contrast, information-theoretic methods, which treat the control and filtering systems as communication channels, are favorable when we investigate the time-domain or channel-orientated trade-offs, e.g., data rate constraint, time-averaged performance cost, and minimum estimation error, of discrete-time, time-varying or nonlinear systems.

	Total and directed information (rates)\footnote{For a communication channel, total information (rate) refers to the mutual information (rate) between the transmitted message and channel output \cite{Martins_TAC_2008, Li_TAC_2013}; directed information (rate) is a causal generalization of input-output mutual information to feedback channels \cite{Massey_ISIT_1990}. In some scenarios, e.g., AWGN channels in \cite[Sec. III]{Kim_TIT_2008} and \cite[Sec. IV]{Weissman_TIT_2013}, these two quantities coincide.} have been identified as an important information-theoretic measure related to the sensitivity trade-offs or Bode-type integrals \cite{Elia_TAC_2004, Martins_TAC_2008, Li_TAC_2013, Fang_2017}, average performance cost \cite{Charalambous_TAC_2017, Kostina_TAC_2019}, and minimum estimation error \cite{Tanaka_TAC_2017}, of both discrete- and continuous-time control and filtering systems. Conventionally, differential entropy rate, Kolmogorov-Bode formula \cite{Kolmogorov_TIT_1956, Yu_TAC_2010}, and their variants are the primal tools and paradigm for studying and calculating these information rates.	Following this paradigm, Zang and Iglesias first showed that the entropy rate difference between the input and output of control channel serves as an information-theoretic interpretation of Bode's integral in discrete-time LTI control systems~\cite{Zang_SCL_2003}. By utilizing the identity between total information and differential entropy, e.g., \eqref{prop33_eq2} in \hyperref[prop33]{Proposition~3.3}, Martins and Dahleh later derived an entropy inequality and an information rate inequality that define the fundamental constraints for the plant's instability rate, a total information rate, in discrete-time control systems; in particular, when the systems are stationary LTI, a Bode-like integral constraint can be restored from their inequalities \cite{Martins_TAC_2008}. Directed information (rate) was used to characterize the fundamental trade-off between communication rate and performance cost of some optimal control problems in \cite{Charalambous_TAC_2017, Kostina_TAC_2019}. For the filtering problem, total information (rate) was utilized to formulate the rate distortion function and design the filter with trade-off concern in \cite{Charalambous_TAC_2014, Tanaka_TAC_2018, Stavrou_SIAMJCO_2018}. Similar properties and results were also found for the total and directed information rates in continuous-time control and filtering systems \cite{Li_TAC_2013, Tanaka_CDC_2017, Wan_arxiv_2022}.

	Apart from the conventional information-theoretic method, the I-MMSE relationships \cite{Guo_TIT_2005}, i.e., the relationships between mutual (or total) information in information theory and (prediction, causal, or non-causal) minimum mean-square estimation error in estimation theory, provides an estimation-based approach to calculate and analyze total information (rate) as a control and filtering trade-off metric. This alternative approach has been utilized to study the trade-offs of continuous-time control and filtering systems, in which differential entropy rate is not well-defined and cumbersome to use. By using Duncan's theorem, an identity between total information and the causal MMSE of input in continuous-time channel, total information rate was studied as a trade-off metric of continuous-time control and filtering systems in \cite{Wan_arxiv_2022}, and rate distortion function was exploited to design an optimal Kalman-Bucy filter with trade-off concern in \cite{Tanaka_CDC_2017, Tanaka_arxiv_2022}. However, few existing paper discussed or invoked the I-MMSE relationships when calculating or analyzing the total information (rates) in discrete-time control and filtering systems. One reason, as above mentioned, is that for the discrete-time channels and systems, we already have some mature and convenient tools, such as entropy rate and Kolmogorov-Bode formula, to perform a pure information-theoretic analysis. More importantly, although the continuous-time I-MMSE relationships, e.g., Duncan's theorem, have been introduced and developed for decades \cite{Duncan_SIAMJAM_1970, Kadota_TIT_1971a}, their discrete-time counterparts previously received less attention. Bucy and Guo et al. briefly commented on the I-MMSE relationship of discrete-time non-feedback channel in \cite{Bucy_IS_1979, Guo_TIT_2005}. Han and Song developed a differential I-MMSE relationship between the total information and the non-causal MMSEs of discrete- and continuous-time additive white Gaussian noise (AWGN) channels subject to feedback and input memory \cite{Han_TIT_2016}. Nevertheless, the more intuitive and practical I-MMSE relationship between the total information and the prediction/causal MMSEs of discrete-time AWGN channel with feedback, which is essential for analyzing the discrete-time control systems, has been due for a long time.

	In this paper, we derive the missing I-MMSE relationship for the discrete-time white Gaussian channels with and without feedback, simultaneously. With this newly derived I-MMSE relationship, and after modeling the discrete-time control and filtering systems into proper communication channels, we then investigate the total information (rate) as a control and filtering trade-off metric and estimate it by using the optimal filtering techniques, which together provide a discrete-time counterpart result complementing our previous analysis of continuous-time systems in \cite{Wan_arxiv_2022}. The main contributions of this paper can be summarized as follows:

	\noindent (C1)\label{cont1} We derive the I-MMSE relationship showing that the total information of discrete-time AWGN channel with feedback is sandwiched by the sums of causal and prediction MMSEs of the channel input (\hyperref[thm24]{Theorem 2.4}) irrespective of the distribution and stationarity of the signals or initial states. This relationship supplements the non-feedback result in \cite[Theorem~9]{Guo_TIT_2005} and serves as a discrete-time counterpart of Duncan's theorem in \cite[Theorem~1]{Kadota_TIT_1971a}.

	\noindent (C2)\label{cont2} By resorting to this new I-MMSE relationship, we derive the fundamental constraints of total information rate, $\bar{I}(E; X_0)$ or $\bar{I}(E; C, X_0)$, in a general nonlinear control setup (\hyperref[thm32]{Theorems 3.2} and \hyperref[thm34]{3.4}) and utilize them to characterize the performance trade-offs of various control systems. By applying these results to discrete-time LTI control systems, a more concise proof is presented to show that $\bar{I}(E; X_0)$ equals the sum of logarithmic unstable poles and serves as an information-theoretic interpretation of the established Bode-type integrals (\hyperref[prop36]{Proposition 3.6} and \hyperref[cor37]{Corollary 3.7}). For the LTV control systems, we then prove that $\bar{I}(E; X_0)$, determined by the antistable dynamics, equals a time-domain Bode's integral that quantifies a time-averaged cost function (\hyperref[prop311]{Proposition~3.11} and \hyperref[cor312]{Corollary~3.12}). For the nonlinear control systems, in which $\bar{I}(E; X_0)$ can be interpreted as a limit of data transmission rate or average cost function, we provide a filtering-based approach to estimate $\bar{I}(E; X_0)$ (\hyperref[prop314]{Proposition~3.14}).

	\noindent (C3) By utilizing the I-MMSE relationship in \hyperref[cont1]{(C1)}, we prove that total information rate $\bar{I}(Y; X_0)$ or $\bar{I}(Y; X_0, W)$ quantifies the lowest achievable or a lower bound of the time-averaged prediction MMSE of noise-free output in a general nonlinear filtering setup (\hyperref[thm42]{Theorems~4.2} and \hyperref[thm44]{4.4}). Subsequently, symmetric to the individual control analyses in (\hyperref[cont2]{C2}), we show that total information rate $\bar{I}(Y; X_0)$, determined by the unstable or anti-stable dynamics, equals the sum of logarithmic unstable poles in the LTI filtering systems (\hyperref[prop45]{Proposition~4.5}), and is identical to a time-domain Bode's integral in the LTV filtering systems (\hyperref[prop47]{Proposition~4.7}). For the nonlinear filtering systems, we propose a nonlinear-filtering-based approach to estimate $\bar{I}(Y; X_0, W)$ (\hyperref[sec44]{Section IV-D}).

	The remainder of this paper is organized as follows. \hyperref[sec2]{Section II} introduces the preliminaries and derives the discrete-time I-MMSE relationship. \hyperref[sec3]{Sections III} and \hyperref[sec4]{IV} use total information (rate) to characterize the performance trade-offs of different control and filtering systems, and \hyperref[sec5]{Section~V} draws the conclusions.

	\noindent \textit{Notations}: For a discrete-time random variable $X_i$, $x_i$ denotes a sample or value of $X_i$. For a discrete-time random process, $X_0^n$ or $X_{0:n}$, from $i=0$ to $n$, $x_{[0:n]}$ stands for a sample path or value of $X_0^n$. The norm of $X_0^n$ is defined by $\|X_0^n\| = (\sum_{i=0}^{n}X_i^\top X_i)^{1/2}$. $x_0^t$ stands for a continuous-time random process from $\tau = 0$ to $t$. For a complex number $c$, $|c|$ denotes its modulus. $\rm e$ is Euler's number. For a matrix $M$, $\det M$ or $|M|$  denotes its determinant. $\mu_i(M)$ denote the singular values of $M$ and are ordered as $\overline{\mu}(M):=\mu_1(M) \geq \mu_2(M) \geq \cdots \geq \mu_r(M) =: \underline{\mu}(M)$. $\|M\|$ denotes the norm of $M$ and equals $\overline{\mu}(M)$. Limit inferior or $\liminf$ is abbreviated as $\underline{\lim}$, and limit superior or $\limsup$ is shortened as $\overline{\lim}$.

	\section{Preliminaries and Problem Formulation}\label{sec2}
	In this section, we first show the definitions and preliminary results of some information-theoretic metrics and discrete-time AWGN channel, and for the first time, derive the I-MMSE relationship between the total information and prediction/causal MMSEs of discrete-time AWGN channel with feedback.
	
	\vspace{-0.2em}
	
	\subsection{Information Theory}
	Throughout this paper, information-theoretic metrics, e.g., differential entropy and mutual information, are defined by their measure-theoretic notions \cite{Polyanskiy_2016, Mehta_TAC_2008}, which apply to the discrete-time processes without probability density function (pdf), and thus are more general than the traditional statistical notions~\cite{Cover_2012}. Differential entropy (rate) measures the amount (or time density) of randomness in a random vector or process. 
	\begin{definition}[Differential Entropy] \label{def21}
		For a discrete-time random process $X_0^n$ with image (or push-forward) measure $\mu^{}_{X}$, its differential entropy $h(X_0^n)$ is defined by
		\begin{equation*}
			h(X_0^n) := \int_{\mathcal{X}} \log\left( \frac{d\mu^{}_{\rm Leb}}{d\mu^{}_X} \right) d\mu^{}_X,
		\end{equation*}
		where $\mathcal{X}$ denotes the Borel $\sigma$-algebra of $X_0^n$, and $\mu^{}_{\rm Leb}$ stands for the Lebesgue measure. In particular, if $X_0^n$ has pdf $p^{}_X$, then $h(X_0^n) := \int_{\mathcal{X}}\log[1/p^{}_X(X_0^n)]  p^{}_X(X_0^n) dX_0^n = \mathbb{E}_{p_{X}^{}}[\log(1 /  p^{}_X(X_0^n))]$. Differential entropy rate of $X_0^n$ is defined by\footref{footnote2}
		\begin{equation*}
			\bar{h}(X) := \uplim_{n\rightarrow\infty} h(X_0^n) / (n+1).
		\end{equation*}
	\end{definition}
	\noindent If $\mu_x$ is not absolutely continuous w.r.t. $\mu^{}_{\rm Leb}$, i.e., $\mu_x \not\ll \mu^{}_{\rm Leb}$, we let $h(X_0^n) = -\infty$ by convention. Mutual information (rate) quantifies the dependency between two random objects and the reliable data transmission (rate) of communication channel.
	\begin{definition}[Mutual Information] \label{def22}
		For discrete-time random processes $X_0^n$ and $Y_0^n$ with image measures $\mu^{}_X$ and $\mu^{}_Y$, the mutual information between $X_0^n$ and $Y_0^n$ is defined by\footnote{\label{footnote2}The definitions of entropy rate and mutual information rate adopted in this paper are allowed to be infinity and aligned with the definitions in \cite{Martins_TAC_2007, Martins_TAC_2008}. The boundedness or existence of these stationary rates has been studied in \cite{Miao_SP_2020, Vu_NC_2009, Carlos_Entropy_2019}. For example, when $X_0^n$ and $Y_0^n$ are stationary Gaussian or jointly stationary, $\bar{h}(X)$ and $\bar{I}(X;Y)$ are bounded, and one can replace $\uplim$ by $\lim$ in their definitions. To focus on the core problem, in this paper, we assume and mainly consider the scenarios when these stationary rates are bounded.}
		\begin{equation*}
			I(X_0^n; Y_0^n) := \int_{\mathcal{X} \times \mathcal{Y}} \log \frac{d\mu^{}_{XY}}{d(\mu^{}_X \times \mu^{}_Y)} d\mu^{}_{XY}\allowdisplaybreaks,
		\end{equation*}
		where $\mathcal{X}$ and $\mathcal{Y}$ respectively denote the Borel $\sigma$-algebras of $X_0^n$ and $Y_0^n$, and $d\mu^{}_{XY} / d(\mu^{}_X \times \mu^{}_Y)$ is the Radon-Nikodym derivative between joint measure $\mu_{XY}$ and product measure $\mu_X \times \mu_Y$. In particular, if $X_0^n$ and $Y_0^n$ have pdfs $p_X^{}$ and $p_Y^{}$, then $I(X_0^n; Y_0^n) := \int_{\mathcal{X} \times \mathcal{Y}} \log [p^{}_{X\times Y}(X_0^n, Y_0^n) / (p_X^{}(X_0^n) \times p_Y^{}(Y_0^n)) ] \  dX_0^n dY_0^n$. Mutual information rate between $X_0^n$ and $Y_0^n$ is defined by
		\begin{equation*}
			\bar{I}(X;Y) := \uplim_{n\rightarrow \infty} {I(X_0^n; Y_0^n)} / (n+1). \allowdisplaybreaks
		\end{equation*}
	\end{definition}
	\noindent Properties of differential entropy and mutual information, e.g., invariance and maximum of differential entropy, data processing inequality, chain rule, symmetry, and non-negativity of mutual information, are not given in this paper. Interested readers are referred to \cite{Cover_2012, Li_TAC_2013, Polyanskiy_2016, Fang_2017, Wan_SCL_2019} and references therein for their rigorous definitions and proofs.

	\subsection{Discrete-Time Gaussian Channel}\label{sec22}
	Consider a discrete-time additive white Gaussian noise or AWGN channel with or without feedback described by \hyperref[fig1]{Fig.~1},
	\begin{figure}[H] \vspace{-0.5em}
		\centerline{\includegraphics[width=0.75\columnwidth]{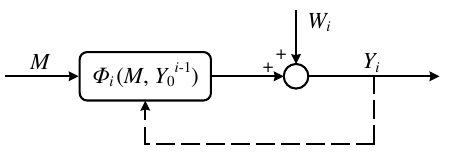}}
		\caption{Discrete-time additive Gaussian channel.} \vspace{-0.5em}
		\label{fig1}
	\end{figure}
	\noindent and the following difference equation
	\begin{equation}\label{Discrete_GC}
		Y_{i} =  \varPhi_i(M, Y_0^{i-1}) + W_i,
	\end{equation}
	\noindent where $Y_i$ and $Y_0^{i-1}$ respectively denote the channel output at time $i$ and the sample path of channel output from $0$ to $i-1$; the transmitted message $M$ can be a random variable or process; $\varPhi_i(M, Y_0^{i-1})$ stands for the channel input process or coding function of $M$, and $W_i$ is a white Gaussian channel noise independent from $M$. To describe a non-feedback channel, we replace the input function $\varPhi_i(M, Y_0^{i-1})$ in \eqref{Discrete_GC} with $\varPhi_i(M)$. Compared with the analysis of continuous-time channels performed on the infinite-dimensional spaces \cite{Kadota_TIT_1971a, Ihara_1993, Wan_arxiv_2022}, the convenience of analyzing the discrete-time channel \eqref{Discrete_GC} is that we can start from a finite-dimensional or -horizon problem with $i = 0, 1, \cdots, n$ and then generalize the results to infinite-horizon by taking the limit $n\rightarrow\infty$ \cite{Guo_2004}.

	Before we introduce the discrete-time I-MMSE relationship, some estimation performance metrics of channel \eqref{Discrete_GC} are defined as follows. Let the prior estimate of the input process be $\hat{\varPhi}^{-}_i := \mathbb{E}[\varPhi_i | Y_0^{i-1}]$. The one-step prediction MMSE of the input process $\varPhi_i$ is then defined by 
	\begin{equation}\label{pmmse}
		{\rm pmmse}(\varPhi_i) := \mathbb{E}[ (\varPhi_i - \hat{\varPhi}_i^{-})^\top(\varPhi_i - \hat{\varPhi}_i^{-})  ].
	\end{equation}
	\noindent Let the posterior estimate of the input process be $\hat{\varPhi}_i := \mathbb{E}[\varPhi_i|Y_0^i]$. The causal filtering MMSE of $\varPhi_i$ is denoted by 
	\begin{equation}\label{cmmse}
		{\rm cmmse}(\varPhi_i) := \mathbb{E}[ (\varPhi_i - \hat{\varPhi}_i)^\top(\varPhi_i - \hat{\varPhi}_i)  ].
	\end{equation}
	\noindent The following continuous-time I-MMSE relationship, which is also known as the Duncan's theorem \cite{Duncan_SIAMJAM_1970}, will be used to derive its discrete-time counterpart. 
	\begin{lemma}\label{lem23}
		For a continuous-time additive Gaussian channel $dy^{}_t = \sqrt{\rm snr} \ \phi(t, m, y_0^t) dt + dw^{}_t$ with transmitted message $m$, signal-to-noise ratio ${\rm snr} > 0$, channel input process $\phi^{}_t(m, y_0^t)$ or $\phi^{}_t(m)$ satisfying $\mathbb{E}[\phi_t^\top \phi_t^{}] < \infty$, channel output $y^{}_t$, and channel noise  $w_t$ -- a standard Brownian motion, the directed information $I(\phi_0^t \rightarrow y_0^t)$ and total information $I(m; y_0^t)$ satisfy
		\begin{equation}\label{lem23_eq1}
			I(\phi_0^t \rightarrow y_0^t) = I(m; y_0^t) = \frac{\rm snr}{2} \int_{0}^{t} {\rm cmmse}(\phi_\tau, {\rm snr}) \ d\tau,
		\end{equation}
		where ${\rm cmmse}(\phi_\tau, {\rm snr}) :=  \mathbb{E}[(\phi_\tau - \hat{\phi}_\tau)^\top (\phi_\tau - \hat{\phi}_\tau)]$ denotes the continuous-time causal MMSE, and $\hat{\phi}^{}_\tau := \mathbb{E}[\phi_\tau^{} | y_0^\tau]$ is the posterior estimate of channel input. 
	\end{lemma}
	\noindent Directed information, as a generalization of input-output mutual information to random objects obeying causal relation, and total information are widely used to analyze the control and communication problems with feedback. However, since total information can grant more flexibility and convenience in our later derivations, it will be favored throughout this paper. The second equality in \eqref{lem23_eq1} reveals the continuous-time I-MMSE relationship or Duncan's theorem, whose proof is available in \cite{Kadota_TIT_1971a, Guo_TIT_2005}. Based on \hyperref[lem23]{Lemma~2.3}, the following theorem shows a discrete-time I-MMSE relationship of Gaussian channel \eqref{Discrete_GC}. 
	\begin{theorem}\label{thm24}
		For the discrete-time additive white Gaussian channel with or without feedback described by \eqref{Discrete_GC} and \hyperref[fig1]{Fig.~1}, when the input process satisfies $\mathbb{E}[\varPhi_i^\top \varPhi^{}_i] < \infty$, $i = 0, \cdots, n$, the total information $I(M; Y_0^n)$ in \eqref{Discrete_GC} is subject to
		\begin{equation}\label{thm24_eq1}
			\frac{1}{2}  \sum_{i=0}^{n} {\rm cmmse}(\varPhi_i) \leq I(M; Y_0^n) \leq \frac{1}{2}  \sum_{i=0}^{n} {\rm pmmse}(\varPhi_i).
		\end{equation}
	\end{theorem}
	\begin{proof}[\rm\bf Proof]
		When \eqref{Discrete_GC} represents a non-feedback channel, inequalities \eqref{thm24_eq1} can be proved by following Theorem~9 in \cite{Guo_TIT_2005} and the identities $I(\varPhi_0^n; Y_0^n) = I(M; Y_0^n) = I(\varPhi_0^n \rightarrow Y_0^n)$ of the non-feedback Gaussian channel presented in \cite{Massey_ISIT_1990, Kim_TIT_2008}, where the discrete-time directed information is defined by $I(\varPhi_0^n \rightarrow  Y_0^n) \break :=  \sum_{i=0}^{n}I(\varPhi_0^i; Y_i|Y_0^{i-1})$. However, it is not trivial to prove that the I-MMSE relationship \eqref{thm24_eq1} can carry over to the discrete-time AWGN channel with feedback, since several fundamental paradigms, such as the SNR-incremental channel, differential I-MMSE relationship along with the related results and proofs in \cite{Guo_TIT_2005}, as well as some information-theoretic measures, e.g., input-output mutual information $I(\varPhi_0^n; Y_0^n)$ in \eqref{Discrete_GC}, are either invalid or not well-defined in the presence of feedback.

		In the following, motivated by \cite{Guo_2004}, we are to prove that inequalities \eqref{thm24_eq1} hold in the discrete-time  AWGN channel with feedback by using the continuous-time I-MMSE relationship or Duncan's theorem in \hyperref[lem23]{Lemma~2.3}. Consider the following continuous-time channel with piecewise-constant message $m$ and input $\phi^{}_{t}$ as an equivalent of the discrete-time channel \eqref{Discrete_GC}
		\begin{equation}\label{thm24_eq2}
			dy^{}_t = \sqrt{\rm snr} \ \phi^{}_{t}(m, y_0^t) dt + dw^{}_t, \quad t \in [0, \infty).
		\end{equation}
		\noindent To align with the expression of \eqref{Discrete_GC}, we assume ${\rm snr} = 1$ in \eqref{thm24_eq2} by default. On the time interval $t \in (i-1, i]$, or more precisely $t \in (t_{i-1}, t_i]$, the piecewise-constant message $m$ is equal to the value of random variable/vector $M$ in \eqref{Discrete_GC} at time $i$, and the input to the continuous-time channel \eqref{thm24_eq2} equals the value of random vector/process $\varPhi_i$ in \eqref{Discrete_GC} at time $i$, i.e., $\phi_t(m, y_0^t) = \varPhi_i(M, Y_0^{i-1})$ if $t \in (i-1, i]$. At time step $t = i$, we further have the equality $w_t = W_i$, and thus $y_t = Y_i$ in \eqref{Discrete_GC} and \eqref{thm24_eq2}. Consequently, the following relationship of total information in \eqref{Discrete_GC} and \eqref{thm24_eq2} can be established 
		\begin{equation}\label{thm24_eq3}
			I(M; Y_0^n) = I(M; y_0^n) = I(m; y_0^n),
		\end{equation}
		\noindent where the first equality uses the fact that $Y_0^n$, or the samples of continuous-time process $y_0^n$ at natural numbers, are sufficient statistics for message $M$; the second equality follows the fact that discrete-time message $M$ and piecewise-constant message $m$ have the same randomness, and $I(M; Y_0^n)$ also equals the directed information $I(\varPhi_0^n \rightarrow Y_0^n)$ by \cite{Kim_TIT_2008}.

		Since the causal MMSEs of the discrete-time channel \eqref{Discrete_GC} and continuous-time channel \eqref{thm24_eq2} are identical at time $t = i$, we denote them collectively by ${\rm cmmse}(\varPhi_i, {\rm snr})$ in the following. Meanwhile, when $t\in(i-1, i]$, since the filtration generated by $y_0^i$ or $Y_0^i$ contains more information about $\phi_i$ or $\varPhi_i$ than the filtration generated by $y_0^t$, which in turn contains more information about $\phi_i$ or $\varPhi_i$ than $y_0^{i-1}$ or $Y_0^{i-1}$, we have
		\begin{equation}\label{thm24_eq4}
			{\rm cmmse}(\varPhi_i, {\rm snr}) \leq {\rm cmmse}(\phi_t, {\rm snr}) \leq {\rm pmmse}(\varPhi_i, {\rm snr}),
		\end{equation}
		\noindent where ${\rm pmmse}(\varPhi_i, {\rm snr})$ and ${\rm cmmse}(\varPhi_i, {\rm snr})$ follow the definitions in \eqref{pmmse} and $\eqref{cmmse}$ when ${\rm snr} = 1$, and the continuous-time causal MMSE ${\rm cmmse}(\phi_t, {\rm snr})$ was defined in \eqref{lem23_eq1}. Letting ${\rm snr} = 1$ in \eqref{thm24_eq2}, and integrating \eqref{thm24_eq4} over $t$ from $0$ to $n$, we obtain \eqref{thm24_eq1}, which follow the facts that $I(M;  Y_0^n)  =   ({\rm snr} /  {2}) \break \int_{0}^{n} {\rm cmmse}(\phi_t^{}, {\rm snr}) dt$ by \hyperref[lem23]{Lemma 2.3} and \eqref{thm24_eq3}, $\int_{0}^{n} {\rm cmmse}(\varPhi_i, \break {\rm snr}) dt  =  \sum_{i=0}^{n} {\rm cmmse}(\varPhi_i, {\rm snr})$, and $\int_{0}^{n} {\rm pmmse}(\varPhi_i,  {\rm snr}) dt  = \sum_{i=0}^{n} {\rm pmmse}(\varPhi_i, {\rm snr})$. 
	\end{proof}	
	\noindent \hyperref[thm24]{Theorem 2.4} develops an I-MMSE relationship between the total information and the prediction/causal MMSEs of discrete-time AWGN channels with and without feedback. This result not only serves as a discrete-time counterpart of the Duncan's theorem in \hyperref[lem23]{Lemma 2.3}, but provides a rigorous approach to generalize \cite[Theorem 9]{Guo_TIT_2005} into the feedback scenario. Similar to the preceding I-MMSE relationships, \hyperref[thm24]{Theorem 2.4} holds regardless of the distribution and stationarity of the signals or initial states in \hyperref[fig1]{Fig. 1}. Unlike Duncan's theorem, which is an equality condition, \hyperref[thm24]{Theorem 2.4} establishes a sandwich inequality between $I(M; Y_0^n)$ and the causal/prediction MMSEs. Although there does exist an equality and differential I-MMSE relationship between $I(M; Y_0^n)$ and the sum of a correctional term and non-causal smoothing MMSEs\footnote{The non-causal MMSEs are defined by ${\rm nmmse}(\varPhi^{}_i) :=  \mathbb{E}[ (\varPhi^{}_i   -  \hat{\varPhi}^{+}_i)^\top \cdot \break  (\varPhi^{}_i - \hat{\varPhi}^{+}_i)]$ with $i = 0,\cdots, n$ and smoothed estimate $\hat{\varPhi}^{+}_i := \mathbb{E}[\varPhi^{}_i|Y_0^n]$.} in the discrete-time AWGN channels \cite{Han_TIT_2016}, since the differential relation is not as straightforward as \eqref{thm24_eq1}, and non-causal MMSEs are not always attainable in practice, we will not cover this equality I-MMSE relationship in this paper. In the following sections, by utilizing \hyperref[thm24]{Theorem~2.4}, we are to interpret how total information (rate) serve as a trade-off metric in capturing fundamental limitations of control and filtering systems.

	\section{Discrete-Time Control Trade-offs}\label{sec3}
	To characterize the performance limits of control systems, we first model the general discrete-time control systems into an additive white Gaussian channel with feedback. By resorting to \hyperref[thm24]{Theorem 2.4} and optimal estimation theory, total information rate and its sandwich bounds are then utilized to capture the fundamental trade-offs of various control systems.
	
	\vspace{-0.2em}
	
	\subsection{General Control Systems and Trade-offs}\label{sec31}
	Consider the general discrete-time feedback control system illustrated by the  diagrams in \hyperref[fig2]{Fig.~2},
	\begin{figure}[H] \vspace{-0.65em}
		\centerline{\includegraphics[width=0.93\columnwidth]{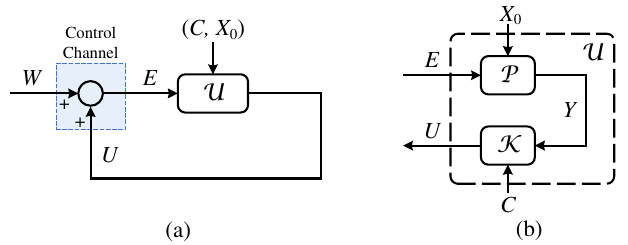}}
		\caption{Configuration of a general control system. (a) Block diagram of control system. (b) Block diagram of control input process $\mathcal{U}$.}\label{fig2}
	\end{figure}\vspace{-0.65em}
	\noindent where $\mathcal{U}$ stands for the lumped control process or the input coding function of ``message" $(C, X_0)$, which can be further decomposed into a plant model $\mathcal{P}$ and a control mapping $\mathcal{K}$; $Y$ and $X_0$ are the measured output and initial states of $\mathcal{P}$; $C$ is the command signal or control noise of $\mathcal{K}$; $U$ denotes the control input; external disturbance or channel noise $W$, independent from $(C, X_0)$, is additive white Gaussian or memory-less, and $E$ stands for the error signal.\footnote{We follow the general control system setup in \cite{Ishii_SCL_2011, Wan_SCL_2019} when the channel noise $W$ is added to the control signal $U$. By switching the roles of $\mathcal{P}$ and $\mathcal{K}$ in \hyperref[fig2]{Fig.~2}, we can also analyze the case when the noise $W$ is imposed on the measured output $Y$. The analysis procedures and results of these two scenarios are symmetric, while the only difference is in the notations.} We use the following stochastic difference equation to describe the plant $\mathcal{P}$: 
	\begin{equation}\label{Ctrl_Plant}
		\begin{split}
			X_{i+1} & = f_i(X_i) + b_i(X_i) E_i \allowdisplaybreaks\\
			Y_{i} & = h_{i}(X_i) ,
		\end{split}
	\end{equation}
	\noindent where $X_i$, $E_i$, and $Y_i$ are respectively the internal states of plant $\mathcal{P}$, error signal, and measured output at time $i$. The control mapping $\mathcal{K}$ can be realized by either the following causal stabilizing function with memory
	\begin{equation}\label{Ctrl}
		U_i = g^{}_i(Y_0^i, C_0^i)
	\end{equation}	
	\noindent or the memory-less function $U_i = g_i(Y_i ,C_i)$. Combining \eqref{Ctrl_Plant} and \eqref{Ctrl}, we can describe the control channel in \hyperref[fig2]{Fig.~2} and the error signal $E_i$ in \eqref{Ctrl_Plant} by
	\begin{equation}\label{Ctrl_Channel}
		E_i = U_i(E_0^{i-1}, C_0^i, X^{}_0) + W_i ,
	\end{equation}
	\noindent where input function $U_i(E_0^{i-1}, C_0^i, X^{}_0) = g_i(Y_0^i, C_0^i)$, and external disturbance or channel noise $W_i \sim \mathcal{N}(0, I)$. Meanwhile, we postulate the following assumption on the control system and channel depicted in \eqref{Ctrl_Plant}-\eqref{Ctrl_Channel}. 
	\vspace{0.5em}
	
	\noindent (A1) \label{ass1} The control system in \eqref{Ctrl_Plant}-\eqref{Ctrl_Channel} is internally mean-square stable, i.e., $\sup_i \mathbb{E}[X_i^\top X^{}_i] < \infty$, and the average power of control input $U_i$ is finite, i.e., $\mathbb{E}[U_i^\top U_i^{}] < \infty, \forall i$. 
	\begin{remark}
		The mean-square stability in \hyperref[ass1]{(A1)} is fundamental and commonly posited in the analyses of stochastic control systems, e.g., \cite{Martins_TAC_2008, Li_TAC_2013}. The average power constraint in \hyperref[ass1]{(A1)} is a preliminary condition for investigating the I-MMSE relationships, since it guarantees the absolute continuity or equivalence between the image measures for calculating total information and mutual information \cite{Ihara_1993, Guo_2004}.
	\end{remark}

	By applying \hyperref[thm24]{Theorem 2.4} to  \eqref{Ctrl_Channel}, we obtain a MMSE-based sandwich bound for the total information (rate) $I(E_0^n; C_0^n, X_0)$ in control channel \eqref{Ctrl_Channel}. Similar to some famous control trade-off metrics, e.g., Bode-like integrals and average entropy cost, $I(E_0^n; C_0^n, X_0)$ characterizes several fundamental trade-offs of control systems varied by the specific form of $\mathcal{U}$.
	\begin{theorem}\label{thm32}
		For the discrete-time control systems described in \eqref{Ctrl_Plant}-\eqref{Ctrl_Channel} and \hyperref[fig2]{Fig.~2}, the total information $I(E_0^n; C_0^n, X^{}_0)$ in \eqref{Ctrl_Channel} is bounded between
		\begin{equation}\label{thm32_eq1}
			\frac{1}{2}  \sum_{i=0}^n {\rm cmmse}(U_i) \leq I(E_0^n; C_0^n, X^{}_0) \leq \frac{1}{2}  \sum_{i=0}^n {\rm pmmse}(U_i),
		\end{equation}
		and the total information rate $\bar{I}(E; C, X_0)$ satisfies\footnote{In particular, when $(E_0^n, C_0^n)$ are stationary Gaussian or jointly stationary, and based on \eqref{thm32_eq1} and \hyperref[def22]{Definition~2.2}, we have $\lim_{n\rightarrow\infty} [2(n+1)]^{-1} \cdot \sum_{i=0}^{n} {\rm cmmse}(U_i)  \leq \bar{I}(E; C, X_0) = \lim_{n\rightarrow \infty} {I(E_0^n; C_0^n, X^{}_0)} / (n+1) \leq  {\lim_{n\rightarrow\infty}}  [2(n+1)]^{-1} \sum_{i=0}^{n} {\rm pmmse}(U_i)$.}
		\begin{align}\label{thm32_eq2}
			\uplim_{n\rightarrow\infty} \sum_{i=0}^{n} \frac{{\rm cmmse}(U_i)}{2(n+1)}    \leq \bar{I}(E; C, X_0) & \leq {\uplim_{n\rightarrow\infty}} \sum_{i=0}^{n} \frac{{\rm pmmse}(U_i)}{2(n+1)} .
		\end{align}
	\end{theorem}	
	\begin{proof}[\rm\bf Proof]
		Inequalities \eqref{thm32_eq1} can be directly proved by applying \hyperref[thm24]{Theorem 2.4} to the discrete-time control channel \eqref{Ctrl_Channel}, wherein $(C_0^i, X_0)$ corresponds to the message $M$ in \eqref{Discrete_GC}. Inequalities \eqref{thm32_eq2} can be derived from \eqref{thm32_eq1} and \hyperref[def22]{Definition 2.2}.		
	\end{proof}	
	\noindent Utilizing the fundamental identities and properties of differential entropy and mutual information, we have the following equality constraints and decomposition for $I(E_0^n; C_0^n, X_0)$. 
	\begin{proposition}\label{prop33}
		For the discrete-time control systems in \eqref{Ctrl_Plant}-\eqref{Ctrl_Channel} and \hyperref[fig2]{Fig.~2}, the total information $I(E_0^n; C_0^n, X_0)$ satisfies
		\begin{equation}\label{prop33_eq1}
			I(E_0^n; C_0^n, X_0) = I(E_0^n; X_0) + I(E_0^n; C_0^n|X_0), \allowdisplaybreaks
		\end{equation}
		and
		\begin{equation}\label{prop33_eq2}
			I(E_0^n; C_0^n, X_0) = h(E_0^n) - h(W_0^n).
		\end{equation}
	\end{proposition}
	\begin{proof}[\rm\bf Proof]
		\eqref{prop33_eq1} is verified by applying the chain rule of mutual information to $I(E_0^n; C_0^n, X_0)$. \eqref{prop33_eq2} is proved by \vspace{-0.5em}
		\begin{align*}
			I(E_0^n; C_0^n, X_0) 
			&\neweq{(a)} h(E_0^n) - \sum_{i=0}^{n} h(E_i | E_0^{i-1}, C_0^n, X_0) \allowdisplaybreaks\\
			&\neweq{(b)} h(E_0^n) - \sum_{i=0}^{n} h(W_i | E_0^{i-1}, C_0^i, X_0) \allowdisplaybreaks\\
			&\neweq{(c)} h(E_0^n) - \sum_{i=0}^{n} h(W_i | W_0^{i-1}, C_0^i, X_0) \allowdisplaybreaks\\
			&\neweq{(d)}  h(E_0^n) - \sum_{i=0}^{n} h(W_i | W_0^{i-1}),
		\end{align*}
		where (a) follows from the identity between mutual information and differential entropy, and the chain rule of differential entropy; (b) relies on \eqref{Ctrl_Channel}, in which $U_i$ is a function of $(E_0^{i-1}, \break  C_0^i,    X_0)$; (c) uses the fact that $(E_0^{i-1}, C_0^i, X_0)$ is a function of $(W_0^{i-1}, C_0^i,   X_0)$, and (d) implies \eqref{prop33_eq2} by the chain rule. 
	\end{proof}	
	\noindent The identities and representation in \hyperref[prop33]{Proposition 3.3}, developed from the pure and conventional information-theoretic relations, have been a pivotal tool for studying the total information (rate) as a control performance trade-off metric in \cite{Li_TAC_2013, Zhao_Auto_2014, Wan_SCL_2019}. For example, the important entropy inequality in \cite[Theorem 4.2]{Martins_TAC_2008}, i.e., $h(E_0^n) \geq I(E_0^n; X_0) + h(W_0^n)$, can be restored from \eqref{prop33_eq1} and \eqref{prop33_eq2}. On the contrary, based on the I-MMSE relationship, \hyperref[thm32]{Theorem 3.2} establishes an inequality constraint and sandwich-type bounds for the total information (rate), which provide an alternative approach for analyzing and estimating this control trade-off metric and will be studied in more detail in the rest of this section.

	By combining \hyperref[thm32]{Theorem 3.2} and \hyperref[prop33]{Proposition 3.3}, we then generalize the existing interpretations and bounds for the total information rate $\bar{I}(E; C, X_0)$. Before we show this extension, the notations of feedback capacity $\mathcal{C}_f$ and channel capacity $\mathcal{C}$ are introduced. For the control channel \eqref{Ctrl_Channel}, we define the (information) feedback capacity\footnote{In this paper, we adopt the mutual information (rate) version of feedback capacity as in \cite{Kadota_TIT_1971b, Ihara_1993}, as opposed to the (operational) feedback capacity, or the supremum of data rate, used in \cite{Kim_TIT_2006, Weissman_TIT_2013}.} $\mathcal{C}_f := \sup_{(C, X_0, U)} \bar{I}(E; C, X_0)$ as the supremum of total information rate over all admissible pairs of the ``message" $(C, X_0)$ and input $U$; for the same channel without feedback, when the power of input signal is limited by $\mathbb{E}[U_i^\top U^{}_i] \leq \rho_i, \forall i$, the channel capacity $\mathcal{C} := \break \sup_{\mathbb{E}[U_i^\top U^{}_i] \leq \rho^{}_i}\bar{I}(U; E)$ is defined as the supremum of input-output mutual information rate. The following theorem shows an extended sandwich bound on $\bar{I}(E; C, X_0)$.
	\begin{theorem}\label{thm34}
		For the discrete-time control systems described in \eqref{Ctrl_Plant}-\eqref{Ctrl_Channel} and \hyperref[fig2]{Fig.~2}, the total information rate $\bar{I}(E; C, X_0)$ in \eqref{Ctrl_Channel} is bounded between 
		\begin{equation}\label{thm34_eq1}
			{\rm CLB} \leq \bar{I}(E; C, X_0) \leq \mathcal{C}_f \leq \uplim_{n\rightarrow\infty} \sum_{i=0}^{n} \frac{{\rm pmmse}(U_i)}{2(n+1)}, \allowdisplaybreaks
		\end{equation}
		where control lower bound ${\rm CLB} =  \max\{ \bar{I}(E; X_0),  \overline{\lim}_{n\rightarrow\infty} \break [2 (n  + 1)]^{-1}   \sum_{i=0}^{n} {\rm cmmse}(U_i) \}$. When $\mathbb{E}[U_i^\top U_i^{}] \leq \rho_i^{}$, $\forall i$, and $\Sigma_{U_i}$ denotes the covariance of $U_i$, the feedback capacity $\mathcal{C}_f = \mathcal{C} =\uplim_{n\rightarrow \infty} \sup_{\mathbb{E}[U_i^\top U^{}_i] \leq \rho^{}_i} [2(n+1)]^{-1} \sum_{i=0}^{n} \log(|\Sigma_{U_i}  +I|)$. 
	\end{theorem} 
	\begin{proof}[\rm\bf Proof]
		We first derive the control lower bound $\rm CLB$ in \eqref{thm34_eq1}. By the non-negativity of mutual information $I(E_0^n; C_0^n | X_0)$ in \eqref{prop33_eq1}, we can imply that $\bar{I}(E; X_0)$ serves as a lower bound for $\bar{I}(E; C, X_0)$. Moreover, from inequalities \eqref{thm32_eq2}, we can tell that $\overline{\lim}_{n\rightarrow\infty} [2(n+1)]^{-1}  \sum_{i=0}^{n} {\rm cmmse}(U_i)$ is also a lower bound for $\bar{I}(E; C, X_0)$. For now, we temporarily let the larger bound be $\rm CLB$, and figure out their order in later subsections with concrete settings. As for the upper bounds, since $\mathcal{C}_f$ is the lowest upper bound for $\bar{I}(E; C, X_0)$, any upper bound for $\bar{I}(E; C, X_0)$ is greater than or equal to $\mathcal{C}_f$. Hence, by invoking \eqref{thm32_eq2}, we have $\bar{I}(E; C, X_0) \leq \mathcal{C}_f \leq \overline{\lim}_{n\rightarrow\infty} [2(n+1)]^{-1}  \sum_{i=0}^{n} {\rm pmmse}(U_i)$ in \eqref{thm34_eq1}.

		It then remains to show that under the power constraint, $\mathbb{E}[U_i^\top U^{}_i] \leq \rho_i, \forall i$, the feedback capacity $\mathcal{C}_f$ equals $\uplim_{n\rightarrow\infty} \break \sup_{\mathbb{E}[U_i^\top U^{}_i] \leq \rho^{}_i} [2(n+1)]^{-1} \sum_{i=0}^{n} \log(|\Sigma_{U_i} + I|)$. Without loss of generality, suppose that $E_i, W_i \in \mathbb{R}^s, \forall i$, in the control channel \eqref{Ctrl_Channel}. We first prove that $\mathcal{C} \leq \mathcal{C}_f \leq  \uplim_{n\rightarrow \infty} \break \sup_{\mathbb{E}[U_i^\top U^{}_i] \leq \rho^{}_i} [2(n+1)]^{-1} \sum_{i=0}^{n} \log(|\Sigma_{U_i} + I|)$, where the first inequality $\mathcal{C} \leq \mathcal{C}_f$ can be implied from their definitions and the identity $I(E_0^n; C_0^n, X_0) = I(E_0^n; U_0^n)$ when \eqref{Ctrl_Channel} is a channel without feedback \cite{Kim_TIT_2008}, and the second equality can be proved by using the facts that $\mathcal{C}_f$ is the lowest upper bound for $\bar{I}(E; C, X_0)$ and 
		\begin{align} \label{thm34_eq2}
			\bar{I}(E; C, X_0) & \neweq{(a)} \bar{h}(E) - \bar{h}(W) \leq \mathcal{C}_f \nonumber\\
			& \hspace{-40pt} \newleq{(b)} \uplim_{n\rightarrow\infty} \sup_{\mathbb{E}[U_i^\top U^{}_i] \leq \rho^{}_i} \frac{\sum_{i=0}^{n}\log[(2\pi {\rm e})^{s} | {\Sigma_{E_i}}|] }{2(n+1)} - \frac{\log [(2\pi {\rm e})^{s}]}{2} \nonumber \\
			&\newleq{(c)} \uplim_{n\rightarrow\infty} \sup_{\mathbb{E}[U_i^\top U^{}_i] \leq \rho^{}_i} \frac{1}{2(n+1)}  \sum_{i=0}^{n} \log(|\Sigma_{U_i} +I|) \nonumber \\
			&\neweq{(d)} \uplim_{n\rightarrow \infty} \frac{1}{2(n+1)} \sum_{i=0}^{n} \log \left( 1 + \rho^{}_i \right),
		\end{align}
		\noindent where (a) follows from \eqref{prop33_eq2} and \hyperref[def21]{Definitions 2.1} and \hyperref[def22]{2.2}; (b) uses the property $h(E_0^n) \leq \sum_{i=0}^{n}h(E_i)$ with equality attained if and only if $\{E_i\}_{i=0}^n$ are mutually independent, maximum entropy condition, and $W_i \sim \mathcal{N}(0, I)$; (c) relies on the fact that $\Sigma_{E_i} = \Sigma_{U_i} + \Sigma_{W_i}$, where $\Sigma_{W_i} = I$, and (d) can be obtained when \eqref{Ctrl_Channel} is one-dimensional, i.e., $s=1$. To prove the other direction, consider a special scenario when \eqref{Ctrl_Channel} is a Gaussian channel without feedback and $U_i \sim \mathcal{N}(0, \Sigma_{U_i})$ are mutually independent for $i = 0$ to $n$. Since the channel capacity under this special setting satisfies
		\begin{equation}\label{thm34_eq3}
			\mathcal{C} \geq \bar{I}(E; U) = \uplim_{n\rightarrow\infty} \sup_{\mathbb{E}[U_i^\top U^{}_i] \leq \rho^{}_i} [2(n+1)]^{-1}  \sum_{i=0}^{n} \log(|\Sigma_{U_i} + I|) ,
		\end{equation}
		\noindent by combining \eqref{thm34_eq2} and \eqref{thm34_eq3}, we prove the last statement in \hyperref[thm34]{Theorem 3.4}. This completes the proof.
	\end{proof}
	\noindent By connecting the MMSE-based bounds in \hyperref[thm32]{Theorem 3.2} with the quantities in \hyperref[prop33]{Proposition 3.3}, \hyperref[thm34]{Theorems 3.4} offers a more comprehensive description on $\bar{I}(E; C, X_0)$. Depending on the specific forms of the plant $\mathcal{P}$ and controller $\mathcal{K}$ in \hyperref[fig2]{Fig. 2}, the total information rate $\bar{I}(E; C, X_0)$ and the plant's instability rate $\bar{I}(E; X_0)$ in \eqref{thm34_eq1} capture various control trade-offs. For example, as later subsections demonstrate, when $\mathcal{P}$ is linear, $\bar{I}(E; X_0)$ serves as an information-theoretic interpretation of the Bode-type integrals (\hyperref[prop36]{Proposition 3.6}) or the average risk-sensitive cost function (\hyperref[prop311]{Proposition 3.11}), and $\bar{I}(E; C, X_0)$ is closely related to the data-rate constraint and rate-distortion trade-offs considered in \cite{Nair_PIEEE_2007, Tanaka_TAC_2018}. Typically, similar to the trade-off properties of aforementioned metrics, the larger the information rates and the bounds in \eqref{thm34_eq1} are, the rougher the control limitations tend to be, e.g., when $\bar{I}(E; C, X_0)$ and its associated bounds increase in a linear control system, the noise sensitivity property or the optimal performance cost of the same system usually deteriorates. In the following subsections, by using \hyperref[thm34]{Theorem 3.4}, $\bar{I}(E; C, X_0)$ or $\bar{I}(E; X_0)$ will be scrutinized as a control trade-off metric with more concrete setups, in which explicit expressions and the order of the bounds in \eqref{thm34_eq1} are attainable. 	  		 
	\begin{remark}
		\hyperref[thm32]{Theorems 3.2}, \ref{thm34}, and \hyperref[prop33]{Proposition 3.3} also hold for the control channel \eqref{Ctrl_Plant}-\eqref{Ctrl_Channel} without the command signal $C_0^n$, in which the total information rate $\bar{I}(E; C, X_0)$ coincides with the instability rate $\bar{I}(E; X_0)$ in \eqref{prop33_eq1} and \eqref{thm34_eq1}.
	\end{remark}

	\subsection{Control Trade-offs in LTI Systems}
	When both the plant model $\mathcal{P}$ and control mapping $\mathcal{K}$ in \hyperref[fig2]{Fig.~2} are LTI, we use the following equations to describe the dynamics of $\mathcal{P}$ or the augmented dynamics of $\mathcal{P}$ and $\mathcal{K}$: 
	\begin{equation}\label{LTI_Ctrl_Plant}
		\begin{split}
			X_{i+1} &= AX_i + BE_i\\
			Y_i &= HX_i
		\end{split}
	\end{equation}
	\noindent such that the control input signal can be expressed as
	\begin{equation}\label{LTI_Ctrl}
		U_i = G  X_i,
	\end{equation}
	\noindent where $A, B$, $H$ and $G$ are time-invariant matrices of proper dimensions, and $X_i$ is the internal states of $\mathcal{P}$ or the augmented states of $\mathcal{P}$ and $\mathcal{K}$. Equations \eqref{LTI_Ctrl_Plant} and \eqref{LTI_Ctrl} are general enough to depict most  scenarios when $\mathcal{K}$ is an LTI controller, e.g., state feedback, output feedback, and observer-based feedback. Consequently, the control channel can be described by
	\begin{equation}\label{LTI_Ctrl_Channel}
		E_i = U_i + W_i = GX_i + W_i,
	\end{equation}
	where $W_i\sim\mathcal{N}(0, I)$ denotes a white Gaussian noise. By resorting to \hyperref[thm34]{Theorem 3.4}, we have the following result on the total information rate $\bar{I}(E; X_0)$ of LTI control systems or the feedback channel depicted in \eqref{LTI_Ctrl_Plant}-\eqref{LTI_Ctrl_Channel}.	 
	\begin{proposition}\label{prop36}
		When the discrete-time LTI control system described by \eqref{LTI_Ctrl_Plant}-\eqref{LTI_Ctrl_Channel} and \hyperref[fig2]{Fig. 2} is internally mean-square stable, the total information rate $\bar{I}(E; X_0)$ is bounded between 	
\begin{equation}\label{prop36_eq1}
			\uplim_{n\rightarrow\infty} \sum_{i=0}^{n} \frac{{\rm cmmse}(U_i)}{2(n+1)}  \leq   \bar{I}(E; X_0) \leq   \uplim_{n\rightarrow\infty} \sum_{i=0}^{n}  \frac{{\rm pmmse}(U_i)}{2(n+1)} ,
\end{equation}
where $\bar{I}(E; X_0) = \sum_{j} \log |\lambda^{+}_j(A)|$ with $\lambda_j^{+}(A)$ denoting the eigenvalues of matrix $A$ with positive real parts, i.e., open-loop unstable poles of $\mathcal{U}$ in \hyperref[fig2]{Fig.~2}.
	\end{proposition}
	\begin{proof}[\rm\bf Proof]  
		See \hyperref[appA]{Appendix A} for the proof of \hyperref[prop36]{Proposition 3.6} and the calculation of ${\rm cmmse}(U_i)$ and ${\rm pmmse}(U_i)$ in \eqref{prop36_eq1}. Based on the information and estimation theories, the direct approach for proving $\bar{I}(E; X_0) = \sum_{j} \log |\lambda^{+}_j(A)|$ in \hyperref[appA]{Appendix A} is more explicit and concise than the existing methods. Moreover, a numerical example verifying the sandwich bounds in \eqref{prop36_eq1} is also given in \hyperref[appA]{Appendix A}. 
	\end{proof}	
	\noindent \hyperref[prop36]{Proposition 3.6} shows an equality condition and the sandwich bounds, based on the MMSEs of $U_i$, for the total information rate or plant's instability rate $\bar{I}(E; X_0)$ subject to \eqref{LTI_Ctrl_Plant}-\eqref{LTI_Ctrl_Channel}. The equality, $\bar{I}(E; X_0) = \sum_{j} \log |\lambda^{+}_j(A)|$, suggests that the total information rate can be used as an information-theoretic interpretation or substitute to the famous Bode's and Bode-like integrals \cite{Freudenberg_1985, Martins_TAC_2008}, which are equal or bounded below by $\sum_{j} \log |\lambda^{+}_j(A)|$. The sandwich bounds in \eqref{prop36_eq1} supplement this control trade-off property or metric by offering an alternative lower bound and an estimation method for $\bar{I}(E; X_0)$.

	For completeness, we briefly discuss the scenario when the LTI plant \eqref{LTI_Ctrl_Plant} is stabilized by the nonlinear control mapping
	\begin{equation}\label{Nonlinear_Ctrl}
		U_i = g^{}_i(Y_0^i),
	\end{equation}
	\noindent which results in the following control channel
	\begin{equation}\label{LTI_Nonlinear_Channel}
		E_i = U_i(E_0^{i-1}, X_0) + W_i.
	\end{equation}
	\noindent Since the control mapping and channel are now nonlinear, the MMSE-based bounds of total information rate need to be calculated by using the nonlinear filtering technique, which will be discussed later in \hyperref[sec34]{Section III-D}. By applying \hyperref[thm34]{Theorem~3.4} to \eqref{LTI_Nonlinear_Channel}, we have the following corollary on $\bar{I}(E; X_0)$. 
	\begin{corollary}\label{cor37}
		When the discrete-time LTI plant \eqref{LTI_Ctrl_Plant} under nonlinear control mapping \eqref{Nonlinear_Ctrl}, as depicted in \hyperref[fig2]{Fig.~2}, is internally mean-square stable, the total information rate $\bar{I}(E; X_0)$ in \eqref{LTI_Nonlinear_Channel} is bounded between
		\begin{equation}\label{cor37_eq1}
			{\rm CLB}_{\rm LTI} \leq \bar{I}(E; X_0) \leq \uplim_{n\rightarrow\infty}  \sum_{i=0}^{n} \frac{{\rm pmmse}(U_i)}{2(n+1)},
		\end{equation}
		where ${\rm CLB}_{\rm LTI} = \max\{  \sum_j \log |\lambda^{+}_j(A)|,   \uplim_{n\rightarrow\infty} [2(n+1)]^{-1} \cdot \sum_{i=0}^{n} {\rm cmmse}(U_i) \}$, with $\lambda_j^{+}(A)$ denoting the eigenvalues of matrix $A$ with positive real parts.
	\end{corollary}
	\begin{proof}[\rm\bf Proof]
		Applying \hyperref[thm32]{Theorem~3.2} to the control channel \eqref{LTI_Nonlinear_Channel}, we derive the MMSE-based lower and upper bounds in \eqref{cor37_eq1}. By applying \cite[Lemma~4.1]{Martins_TAC_2007} or the \hyperref[cor312]{Corollary~3.12} of this paper to the control systems depicted in \eqref{LTI_Ctrl_Plant}, \eqref{Nonlinear_Ctrl}, and \eqref{LTI_Nonlinear_Channel}, we can tell that $\sum_j \log|\lambda^{+}_j(A)|$ is also a lower bound of $\bar{I}(E; X_0)$. Since the control mapping \eqref{Nonlinear_Ctrl} is now nonlinear, similar to \hyperref[thm34]{Theorem 3.4}, we assign the larger lower bound to ${\rm CLB}_{\rm LTI}$.
	\end{proof}
	\noindent Thanks to the adoption of the I-MMSE relationships in \hyperref[thm24]{Theorem 2.4} or \ref{thm34}, no restriction is imposed on the stationarity or distribution of the signals in \hyperref[prop36]{Proposition 3.6} and \hyperref[cor37]{Corollary 3.7}. Moreover, if all signals in \hyperref[fig2]{Fig. 2} are stationary Gaussian, \cite{Martins_TAC_2008} reveals that the information rate $\bar{I}(E; X_0)$ is identical to a Bode-like integral that is i) defined by the power spectral density ratio of $E_0^n$ and $W_0^n$, ii) lower-bounded by $\sum_j \log |\lambda_j^{+}(A)|$, and iii) captures the noise sensitivity trade-off property as the classical Bode's integral in \cite{Sung_IJC_1988}. \hyperref[prop36]{Proposition~3.6} and \hyperref[cor37]{Corollary~3.7} also supplement the previous findings in \cite{Martins_TAC_2008} by showing i) $\bar{I}(E; X_0)$ {\it equals} $\sum_j \log |\lambda_j^{+}(A)|$ when the control mapping $\mathcal{K}$ is linear; ii) alternative and MMSE-based sandwich bounds for $\bar{I}(E; X_0)$ when the established bounds, such as $\sum_j \log |\lambda_j^{+}(A)|$ and $\mathcal{C}_f$, are loose or difficult to calculate, and iii) nonlinear controller mapping does not necessarily improve the noise sensitivity trade-off of LTI plant.

\subsection{Control Trade-offs in LTV Systems}
When the plant model $\mathcal{P}$ in \hyperref[fig2]{Fig.~2} is linear time-varying, we use \eqref{LTV_Ctrl_Plant} to describe the dynamics of $\mathcal{P}$, or the augmented dynamics of $\mathcal{P}$ and $\mathcal{K}$ when $\mathcal{K}$ is linear and $U_i$ satisfies \eqref{LTV_Ctrl}:
	\begin{equation}\label{LTV_Ctrl_Plant}
		\begin{split}
			X_{i+1} & = A_iX_i + B_iE_i \\
			Y_i & = H_i X_i,
		\end{split}
	\end{equation}
	\noindent where $A_i$, $B_i$ and $H_i$ are time-varying matrices of proper dimensions. Definitions of some notations related to the LTV system \eqref{LTV_Ctrl_Plant}, such as uniformly exponentially (anti)stable, exponential dichotomy, weighted shift operator and its spectrum, are briefly interpreted in \hyperref[appB]{Appendix B} and references therein. Meanwhile, we postulate the following assumptions on \eqref{LTV_Ctrl_Plant}. 
	
	\vspace{0.3em}
	
	\noindent (A2) \label{ass2} The sequence $\{A_i\}$ admits an exponential dichotomy of rank $m_u$, and the spectrum of the antistable part satisfies $\cup_{j=1}^{l}\{\lambda\in\mathbb{C}: \underline{\kappa}_j \leq |\lambda| \leq \overline{\kappa}_j \}$ with multiplicities $m_1, \cdots, m_l$, where $\underline{\kappa}_j \geq 0$ and $\sum_{j=1}^{l} m_j = m_u$. 
	
	\vspace{0.3em}
	
	\noindent (A3) \label{ass3} The sequence $\{A_i - B_i H_i\}$ is uniformly exponentially stable (UES). 
	
	\begin{remark}\label{rem38}
		(\hyperref[ass2]{A2}) can be regarded as an LTV counterpart of the stable/unstable dichotomy in LTI systems and is extensively postulated and studied in the literature on LTV systems \cite{BenArtzi_IEOT_1991, Iglesias_Auto_2001, Dieci_JDE_2010, Tranninger_CSL_2020}. When \eqref{LTV_Ctrl_Plant} is regular, e.g., time-invariant or periodic, the annuli in (\hyperref[ass2]{A2}) shrink to circles, and the modulus of the spectrum in (\hyperref[ass2]{A2}) coincide with the discrete-time Lyapunov exponents or the modulus of unstable poles. (\hyperref[ass3]{A3}) guarantees the boundedness of sensitivity operator, a state-space representation of the sensitivity function, which is also a necessary requirement of any internally stabilizing controller \cite{Iglesias_Auto_2001}. 	
	\end{remark}
	\noindent The exponential dichotomy in (\hyperref[ass2]{A2}) allows us to define a stability preserving state space transformation that similar to the modal decomposition of LTI systems, separates the stable and antistable parts of $\{A_i\}$. 
	\begin{lemma}[See \cite{BenArtzi_IEOT_1991}]\label{lem39}
		The sequence $\{A_i\}$ in \eqref{LTV_Ctrl_Plant} possesses an exponential dichotomy if and only if there exist a bounded sequence of matrices $\{T_i\}$ with bounded inverse such that 
		\begin{align*}
			\left[\begin{array}{c|c}
				T^{}_{i+1}A^{}_i T_i^{-1} & T_{i+1}B_i \\ \hline
				H^{}_iT^{-1}_i & D_i
			\end{array}  \right] 
			:=
			\left[\begin{array}{cc|c}
				A_s(i) & 0 & B_s(i)\\
				0 & A_u(i) & B_u(i)\\ \hline
				H_s(i) & H_u(i) & D(i)
			\end{array}\right],
		\end{align*}
		where $D_i$ denotes the direct feedthrough matrix, and $\{A_s(i)\}$ and $\{A_u(i)\}$ respectively stand for the stable and antistable parts of $\{A_i\}$.
	\end{lemma}
	\noindent By using the relationship between entropy and risk-sensitive cost functions in \cite{Glover_SCL_1988, Peters_TAC_1999}, we define an LTV Bode's integral $\mathscr{B}$, which gives a time-domain characterization of the sensitivity property of \eqref{LTV_Ctrl_Plant}. 
	\begin{lemma}[See \cite{Iglesias_Auto_2001}]\label{lem310}
		For the discrete-time LTV system \eqref{LTV_Ctrl_Plant} fulfilling (\hyperref[ass2]{A2}) and (\hyperref[ass3]{A3}), the time-domain Bode's integral $\mathscr{B}$ satisfies\footnotemark
		\begin{align*}
			\mathscr{B} = \uplim_{n\rightarrow\infty} \frac{\log \det [ \varPhi_{A_u}(n, 0)]}{n}  
			\geq \sum_{j=1}^{l} m_j \log \underline{\kappa}_j \geq m_u \log \beta, 
		\end{align*}
		where $\varPhi_{A_u}(n, 0) = \prod_{i=0}^{n-1}A_u(i)$ stands for the discrete-time state transition matrix of $\{A_u(i)\}$ from $i = 0$ to $n$, and $\alpha$ and $\beta > 1$ are positive constants, independent from $n$, such that $\underline{\mu}(\varPhi_{A_u}(n, 0)) \geq \alpha \beta^{n}$. \footnotetext{Rigorously, the LTV Bode's integral is defined as $\mathscr{B} := \overline{\lim}_{n\rightarrow\infty}(n+1)^{-1} \log \det(E_n^\top S^\top S E_n)$ in \cite{Iglesias_Auto_2001}, where $E_n$ denotes the embedding operator of $X_0^n$ and $S$ is the sensitivity operator of \eqref{LTV_Ctrl_Plant}. Meanwhile, there exist an identity between $\mathscr{B}$ and the average risk-sensitive cost function $\mathscr{B} = \overline{\lim}_{n\rightarrow\infty}-[2(n+1)]^{-1}\log(\mathbb{E}[\exp(\|W_0^n\| / 2 - \|E_0^n\| /2)])$, where $W_0^n$ and $E_0^n$ are the noise and error signals in \hyperref[fig2]{Fig.~2}, and identities $\mathscr{B} = \overline{\lim}_{n\rightarrow\infty}(n+1)^{-1} \sum_{i=0}^{n}\log\det R_i = \overline{\lim}_{n\rightarrow\infty}(n+1)^{-1}\log\det[\varPhi_{A_u}(n+1, 0)]$, where block-diagonal operator $R_i = I + B_u(i)^\top X_u(i+1)B_u(i)$ with $X_u(i)$ and $B_u(i)$ being the antistable parts of $X_i$ and $B_i$.}
	\end{lemma}
	\noindent In particular, if the sequence $\{A_i\}$ in \eqref{LTV_Ctrl_Plant} is open-loop stable, the LTV Bode's integral reduces to $\mathscr{B} = 0$. When the sequence $\{A_u(i)\}$ of unstable part is regular, by \hyperref[rem38]{Remark~3.8}, the LTV Bode's integral takes the form  $\mathscr{B} = \sum_{j=1}^{l} m_j \log \underline{\kappa}_j = \sum_{j=1}^{l} m_j \log \overline{\kappa}_j$.

	With the assumptions and properties above, we now turn to investigate the control trade-offs for the discrete-time LTV systems. Similar to the LTI study, we first consider the scenario when the control mapping $\mathcal{K}$ in \hyperref[fig2]{Fig.~2} is linear, and the control input satisfies 
	\begin{equation}\label{LTV_Ctrl}
		U_i = G_iX_i,
	\end{equation}
	\noindent where $G_i$ is a time-varying feedback gain, and $X_i$ denotes the internal states of $\mathcal{P}$ or the augmented states of $\mathcal{P}$ and $\mathcal{K}$. Subsequently, the control channel subject to the LTV plant \eqref{LTV_Ctrl_Plant} and controller \eqref{LTV_Ctrl} becomes
	\begin{equation}\label{LTV_Ctrl_Channel}
		E_i = U_i + W_i = G_iX_i + W_i.
	\end{equation}
	\noindent By applying \hyperref[thm34]{Theorem~3.4} to \eqref{LTV_Ctrl_Plant}-\eqref{LTV_Ctrl_Channel}, we have the following result on the total information rate $\bar{I}(E; X_0)$ of LTV systems.
	\begin{proposition}\label{prop311}
		When the discrete-time LTV control system subject to \eqref{LTV_Ctrl_Plant}-\eqref{LTV_Ctrl_Channel}, assumptions (\hyperref[ass2]{A2}) and (\hyperref[ass3]{A3}), as \hyperref[fig2]{Fig.~2} shows, is internally mean-square stable, the total information rate $\bar{I}(E; X_0)$ in \eqref{LTV_Ctrl_Channel} is bounded between
		\begin{equation}\label{prop311_eq1}
			{\rm CLB}_{\rm LTV} \leq   \bar{I}(E; X_0) \leq   \uplim_{n\rightarrow\infty} \sum_{i=0}^{n}  \frac{{\rm pmmse}(U_i)}{2(n+1)},
		\end{equation}
		where $\bar{I}(E; X_0) = \overline{\lim}_{n\rightarrow \infty} (n+1)^{-1} \log\det[\varPhi_{A_u}(n+1, 0)]  = \mathscr{B}$, and ${\rm CLB}_{\rm LTV} = \max\{ \sum_{j=1}^l m_j \log \underline{\kappa}_j,  \overline{\lim}_{n\rightarrow\infty}[2(n+1)]^{-1}  \sum_{i=0}^n {\rm cmmse}(U_i) \}$.
	\end{proposition}
	\begin{proof}[\rm\bf Proof]
		By applying \hyperref[thm34]{Theorem 3.4} to the LTV control system and channel depicted by \eqref{LTV_Ctrl_Plant}-\eqref{LTV_Ctrl_Channel}, we can obtain the MMSE-based sandwich bounds in \eqref{prop311_eq1}. Hence, it remains to show that $\bar{I}(E;X_0) = \overline{\lim}_{n\rightarrow\infty}(n+1)^{-1} \log \det \varPhi_{A_u}(n+1, 0) \geq \sum_{j=1}^l m_j \log \underline{\kappa}_j$, where the inequality can be directly implied from \hyperref[lem310]{Lemma 3.10}, and the equality can be verified by respectively proving the statements i)~$\bar{I}(E; X_0) \geq  \overline{\lim}_{n\rightarrow\infty}(n+1)^{-1}  \log  \det \varPhi_{A_u}(n+1, 0)$, and ii)~$\bar{I}(E; X_0) \leq \overline{\lim}_{n\rightarrow\infty} (n+1)^{-1}  \log\det \varPhi_{A_u}(n+1, 0)$.

		To calculate the prior and posterior estimates or estimation error (covariance) of the control input $U_i = G_iX_i$, consider the following LTV filtering problem originated from \eqref{LTV_Ctrl_Plant}-\eqref{LTV_Ctrl_Channel}
		\begin{equation}\label{LTV_Ctrl_Filtering}
			\begin{split}
				X_{i+1} & = (A_i + B_iG_i)X_i + B_i W_i \\
				E_i & = G_i X_i + W_i,
			\end{split}
		\end{equation}
		\noindent where $E_i$ is observable, and $X_i$ is the hidden states to be estimated. By applying the LTV Kalman filter, \cite[Sec. 2.7.2]{Lewis_2017}, to \eqref{LTV_Ctrl_Filtering}, and defining the state prior estimate $\hat{X}_i^{-} := \mathbb{E}[X_i | E_0^{i-1}]$, the non-zero prior error covariance $P_i^{-} := \mathbb{E}[(X_i - \hat{X}_i^{-})(X_i - \hat{X}_i^{-})^\top]$ of the filtering problem \eqref{LTV_Ctrl_Filtering} satisfies
		\begin{equation}\label{LTV_Prediction}
			P_{i}^{-}  = A^{}_{i-1}P_{i-1}^{-}(G^{}_{i-1}P^{-}_{i-1}G_{i-1}^\top + I)^{-1} A_{i-1}^\top.
		\end{equation}
		\noindent Let the posterior estimate be $\hat{X}_i := \mathbb{E}[X_i| E_0^i]$. The posterior error covariance $P_i := \mathbb{E}[(X_i - \hat{X}_i)(X_i - \hat{X}_i)^\top]$ of \eqref{LTV_Ctrl_Filtering} follows
		\begin{equation}\label{LTV_Correction}
			P_{i} = P_{i}^{-}-P_{i}^{-}G^{\top}_{i}(G^{}_{i}P_{i}^{-}G^{\top}_{i}+I)^{-1}G_{i}^{}P_{i}^{-}.
		\end{equation}
		\noindent The evolution of error covariance matrices (not state estimates) of the filtering problem \eqref{LTV_Ctrl_Filtering}, presented in \eqref{LTV_Prediction} and \eqref{LTV_Correction}, are identical to the filtering result of the following LTV system without control input and process noise
		\begin{equation}\label{prop311_eq4}
			\begin{split}
				X_{i+1} & = A_iX_i, \\
				E_i & = G_i X_i + W_i.
			\end{split}
		\end{equation}
		\noindent Since \eqref{prop311_eq4} is under the assumption (\hyperref[ass2]{A2}), we can partition \eqref{prop311_eq4} into stable and antistable parts, as \hyperref[lem39]{Lemma~3.9} shows. When \eqref{LTV_Ctrl_Filtering} or \eqref{prop311_eq4} is either a) uniformly completely reconstructible or b) exponentially stable, by \hyperref[lem46]{Lemma~4.6}, the asymptotic (or steady-state) prior error covariance $\lim_{i' \rightarrow\infty} P^{-}_{i'}$ in \eqref{LTV_Prediction} takes the form of
		\begin{equation}\label{prop311_eq5}
			\lim_{i'\rightarrow\infty}P_{i'}^{-} = \left[\begin{matrix}
				0 & 0\\
				0 & \bar{P}_u^{-}(i)
			\end{matrix}\right],
		\end{equation}
		\noindent where the asymptotic prior error covariance of antistable part $\bar{P}_u^{-}(i) := \lim_{i'\rightarrow\infty}\mathbb{E}[(X_u(i') - \hat{X}_u^{-}(i'))(X_u(i') - \hat{X}_u^{-}(i')  )^\top]$, with antistable modes $X_u(t)$ and their prior estimates $\hat{X}^{-}_u(i)$, satisfies the following RDE (Riccati difference equation)
		\begin{equation}\label{LTV_RDE_Ctrl}
			\bar{P}_{u}^{-}(i+1) = A^{}_u(i) \bar{P}_u^{-}(i) [ G^{}_u(i) \bar{P}_u^{-}(i)G_u^\top(i) + I]^{-1} A_u^\top(i)
		\end{equation}
		\noindent with $\bar{P}_{u}^{-}(i)$ being non-zero, and $A_u(i)$ and $G_u(i)$ respectively being the antistable parts of $A(i)$ and $G(i)$.

		With the filtering result ahead, we now prove statement i) and show that $\bar{I}(E; X_0) \geq \overline{\lim}_{n\rightarrow\infty}(n+1)^{-1}  \log \det  \varPhi_{A_u}  (n+\break 1,  0)  = \overline{\lim}_{n\rightarrow\infty}[2(n+1)]^{-1}  \sum_{i=0}^{n} \log\det[ G^{}_u(i)  \bar{P}_u^{}(i) G_u^\top(i) \break + I]$. The former inequality will be verified later in the proof of \hyperref[cor312]{Corollary 3.12} or \hyperref[appC]{Appendix C} by resorting to the mean-square stability and maximum entropy conditions. To show the latter equality, take $\log \det(\cdot)$ on both sides of \eqref{LTV_RDE_Ctrl}, which gives
		\begin{align}\label{prop311_eq7}
			\log\det A_u(i) & = \frac{1}{2} \Big[ \log \det [G_u(i)\bar{P}_u^{-}(i)G_u^\top(i) + I] \\
			& \hspace{30pt} + \log \det \bar{P}_u^{-}(i+1) - \log \det \bar{P}_u^{-}(i) \Big]. \nonumber
		\end{align}
		\noindent Summing up \eqref{prop311_eq7} from $i = 0$ to $n$, dividing the result by $n+1$, and taking the limit superior as $n\rightarrow\infty$, we then prove statement i). To verify statement ii) and show that $\bar{I}(E; X_0) \leq \overline{\lim}_{n\rightarrow\infty} [2(n+1)]^{-1} \sum_{i=0}^{n} \log\det[ G^{}_u(i)\bar{P}^{}_u(i)G_u^\top(i) + I]  = \overline{\lim}_{n\rightarrow\infty} (n+1)^{-1} \log \det \varPhi_{A_u}(n+1, 0)$, consider the following manipulation on $I(E_0^n; X_0)$,
		\begin{align}\label{prop311_eq8}
			I(E_0^n; X_0) & \neweq{(a)} h(E_0^n) - h(W_0^n) \allowdisplaybreaks\\
			& \neweq{(b)} h(G_0^n (X_0^n - \hat{X}_{0:n}^{-}) + W_0^n) - h(W_0^n) \nonumber \allowdisplaybreaks \\
			& \newleq{(c)} \frac{1}{2} \sum_{i=0}^{n} \log  [ (2\pi {\rm e})^s |G_i^{} P_i^{-}G_i^\top+I|] - \frac{n}{2} \log( 2\pi {\rm e} )^s \nonumber \allowdisplaybreaks \\
			& \neweq{(d)}  \log\det\varPhi_{A_u}(n+1, 0) + \frac{1}{2} \log \det \bar{P}_u^{-}(0) \nonumber \allowdisplaybreaks \\
			& \hspace{102pt} - \frac{1}{2} \log\det \bar{P}_u^{-}(n+1) ,\nonumber
		\end{align}
		\noindent where (a) relies on \eqref{prop33_eq2} in \hyperref[prop33]{Proposition~3.3}; (b) uses the identity \eqref{LTV_Ctrl_Channel} and $h(G_0^n X_0^n + W_0^n) = h(G_0^n(X_0^n - \hat{X}_{0:n}^{-})  + W_0^n)$, since $\hat{X}_{0:n}^{-} := \{\mathbb{E}[X_0^i|E_0^{i-1}]\}_{i=0}^n$ is deterministic; (c) follows from the property $h(X_0^n)  \leq \sum_{i=0}^{n}  h(X_i)$, where equality holds if and only if $\{X_i\}_{i=0}^n$ are mutually independent, and the maximum entropy property with $E_i$ and $W_i \in \mathbb{R}^s$ in \eqref{LTV_Ctrl_Channel}, and (d) invokes \eqref{prop311_eq5} and \eqref{prop311_eq7} in the limit of $i\rightarrow\infty$. 	
		Dividing both sides of \eqref{prop311_eq8} by $n+1$ and taking the limit superior as $n\rightarrow\infty$, statement ii) is proved. Combining the proofs of statements i) and ii), we then verify the identity $\bar{I}(E; X_0) = \overline{\lim}_{n\rightarrow \infty} (n+1)^{-1} \log\det[\varPhi_{A_u}(n+1, 0)]$ in \hyperref[prop311]{Proposition~3.11}. This completes the proof. 		
	\end{proof}	
	\noindent \hyperref[prop311]{Proposition 3.11} reveals that for the LTV systems governed by \eqref{LTV_Ctrl_Plant}-\eqref{LTV_Ctrl_Channel}, total information rate $\bar{I}(E; X_0)$ equals the LTV Bode's integral $\mathscr{B}$ defined in \hyperref[lem310]{Lemma 3.10}. This equality can be regarded as the LTV counterpart of the identity $\bar{I}(E; X_0) = \sum_{j} \log  |\lambda^{+}_j(A)|$ in \hyperref[prop36]{Proposition 3.6}, a technical explanation of which is given in  \hyperref[rem38]{Remarks 3.8} and \ref{rem313}. Meanwhile, compared with the optimization-based proof in \cite{Tanaka_TAC_2018}, which has not been verified in the LTV scenario yet, our proofs of \hyperref[prop36]{Propositions 3.6} and \ref{prop311}, based on the information and estimation theories, are not only more straightforward, but applicable to both the LTI and LTV control systems. Similar to the additional findings${}^{\ref{LTI_foot}}$ in the proof of \hyperref[prop36]{Proposition 3.6}, more interesting observations on the MMSE-based bounds in \eqref{prop311_eq1} can be extracted from the propagation formulas of error covariance matrices in \eqref{LTV_Prediction} and \eqref{LTV_Correction}.\footnote{For example, {\bf a)} $2\sum_{i=0}^{n} \sum_j \log|\lambda_{j}(A_u(i)) | + \log\det P_u^{-}(0) - \log\det \break P_u^{-}(n+1) = \sum_{i=0}^{n} \sum_k (1+\eta^{}_{i, k})$, where $\{\lambda_j(A_u(i))\}$ denote the eigenvalues of $A_u(i)$, and $\{\eta^{}_{i, k}\}$ are the eigenvalues of $G^{}_u(i)P_u^{-}(i)G_u^\top(i)$; {\bf b)} $\lim_{n\rightarrow \infty} [2(n  + 1)] ^{-1} \sum_{i=0}^{n}{\rm cmmse}(U_i)  \leq \overline{\lim}_{n\rightarrow\infty} (n+  1)^{-1}  \sum_{i=0}^{n}  \sum_j \break \log|\lambda_j(A_u(i))|$; {\bf c)} $\lim_{n\rightarrow \infty} [2 (n  + 1)]^{-1} \sum_{i=0}^{n}{\rm pmmse}(U_i) \geq \overline{\lim}_{n\rightarrow\infty} \break  (n+1)^{-1} \sum_{i=0}^{n} \sum_j \log|\lambda_j(A_u(i))|$, where observations b) and c) along with the identity $\bar{I}(E; X_0) = \overline{\lim}_{n\rightarrow \infty} (n+1)^{-1} \log\det[\varPhi_{A_u}(n+1, 0)]$ reaffirm \eqref{prop311_eq1} in \hyperref[prop311]{Proposition 3.11}. The proofs of observations a)-c) can be inferred from the LTI derivations in  \hyperref[appA]{Appendix A}.}

	For completeness, we briefly discuss the scenario when the LTV plant \eqref{LTV_Ctrl_Plant} is stabilized by the nonlinear controller \eqref{Nonlinear_Ctrl}, which can also be modeled into the control channel \eqref{LTI_Nonlinear_Channel}, and is no longer an LTV system due to the non-linearity in \eqref{Nonlinear_Ctrl}. By applying \hyperref[thm34]{Theorem 3.4} to the control system depicted by \eqref{Nonlinear_Ctrl}, \eqref{LTI_Nonlinear_Channel} and \eqref{LTV_Ctrl_Plant}, we have the following finding.
	\begin{corollary}\label{cor312}
		When the discrete-time LTV plant \eqref{LTV_Ctrl_Plant} under nonlinear control mapping \eqref{Nonlinear_Ctrl} in  \hyperref[fig2]{Fig.~2} is internally mean-square stable, the total information rate $\bar{I}(E; X_0)$ satisfies
		\begin{equation}\label{cor312_eq1}
			{\rm CLB}_{\rm LTV} \leq \bar{I}(E; X_0) \leq \uplim_{n\rightarrow\infty}  \sum_{i=0}^{n}  \frac{{\rm pmmse}(U_i)}{2(n+1)},
		\end{equation}
		where ${\rm CLB}_{\rm LTV} \hspace{-2.5pt} = \hspace{-2.5pt} \max\{ \sum_{j=1}^l m_j \log \underline{\kappa}_j, \uplim_{n\rightarrow\infty} [2(n+1)]^{-1} \hspace{-2pt} \cdot \break \sum_{i=0}^{n} {\rm cmmse}(U_i)  \}$.
	\end{corollary}	
	\begin{proof}[\rm\bf Proof] %
		See \hyperref[appC]{Appendix C} or \cite[Cor. 3.16]{Wan_2024} for the proof.
	\end{proof}
	\noindent Since the control mapping \eqref{Nonlinear_Ctrl} is nonlinear, the estimates of input, ${\rm cmmse}(U_i)$ and ${\rm pmmse}(U_i)$, in \eqref{cor312_eq1} can no longer be calculated from the linear Kalman filters adopted in \hyperref[prop36]{Proposition 3.6} or \ref{prop311}. Instead, the nonlinear filtering method, to be discussed in the following subsection, should be used.

	\hyperref[prop311]{Proposition 3.11} and \hyperref[cor312]{Corollary 3.12} connect the total information rate $\bar{I}(E; X_0)$ to the LTV Bode's integral $\mathscr{B}$ and the sum of spectral values $\sum_{j=1}^{l}m_j \log\underline{\kappa}_j$ for the first time. This connection suggests that analogous to the LTV Bode's integral and average risk-sensitive cost function, total information rate also captures the time-domain control trade-off property of the LTV system \eqref{LTV_Ctrl_Plant}, i.e., if the running risk-sensitive or entropy cost is less (or greater) than $\bar{I}(E; X_0)$ on a time interval, the cost outside this interval must be greater (or less). Moreover, with the aid of the I-MMSE relationship, \hyperref[prop311]{Proposition 3.11} and \hyperref[cor312]{Corollary 3.12} do not posit any assumption on the stationarity or distribution of signals or initial states, and the MMSE-based bounds not only provide a sandwich estimate to $\bar{I}(E; X_0)$ but supplement the bounds in \hyperref[thm34]{Theorem 3.4} and \hyperref[lem310]{Lemma 3.10}.
	\begin{remark}\label{rem313}
		The LTV derivations and results in \hyperref[prop311]{Proposition 3.11} and \hyperref[cor312]{Corollary 3.12} cover the LTI ones in \hyperref[prop36]{Proposition 3.6} and \hyperref[cor37]{Corollary 3.7} as a special case. When we apply the LTV results of this section to an LTI system as \eqref{LTI_Ctrl_Plant}, the spectral values $\underline{\kappa}_j$ in the ${\rm CLB}_{\rm LTV}$ of \eqref{prop311_eq1} and \eqref{cor312_eq1} will coincide with the modulus of unstable poles $|\lambda^{+}_j(A)|$ in the ${\rm CLB}_{\rm LTI}$ of \eqref{prop36_eq1} and \eqref{cor37_eq1}, and the time-varying estimation error covariance matrices in \eqref{LTV_Prediction}-\eqref{prop311_eq8} will degenerate to the time-invariant or steady-state matrices and be canceled. Similar relationship also exists when we investigate the total information rates of continuous-time LTI and LTV control systems \cite{Wan_arxiv_2022}.
	\end{remark}

	\subsection{Control Trade-offs in Nonlinear Systems}\label{sec34}
	When both $\mathcal{P}$ and $\mathcal{K}$ in \hyperref[fig2]{Fig.~2} are nonlinear, consider the following nonlinear plant 
	\begin{equation}\label{nonlinear_dym}
		\begin{split}
			X_{i+1} & = f_i(X_i) + b_i(X_i) E_i,\\
			Y_i & = h_i(X_i), 
		\end{split}
	\end{equation}
	stabilized by the nonlinear control mapping 
	\begin{equation}\label{nonlinear_ctrl}
		U_i = g_i(Y_i).
	\end{equation}
	The control channel and error signal $E_i$ in \hyperref[fig2]{Fig.~2} then satisfy
	\begin{equation}\label{nonlinear_ctrl_channel}
		E_i = U_i(E_0^{i-1}, X_0) + W_i = U_i(X_{i-1}) + W_i,
	\end{equation}	
	\noindent where $W_i$ is a white Gaussian process noise, and the notation $U_i$ is abused here to denote proper functions that satisfy $g(Y_i) = U_i(E_0^{i-1}, X_0) = U_i(X_{i-1})$. By applying \hyperref[thm34]{Theorem 3.4} and nonlinear estimation theory to the control system and feedback channel depicted in \eqref{nonlinear_dym}-\eqref{nonlinear_ctrl_channel}, the MMSE-based sandwich bounds of $\bar{I}(E; X_0)$ and their calculation scheme are given in the following proposition and proof. 
	\begin{proposition}\label{prop314}
		When the discrete-time control system in \hyperref[fig2]{Fig. 2}, subject to the nonlinear plant \eqref{nonlinear_dym} and control mapping \eqref{nonlinear_ctrl}, is internally mean-square stable, the total information rate $\bar{I}(E; X_0)$ in \eqref{nonlinear_ctrl_channel} is bounded by
		\begin{equation}\label{prop314_eq1}
			\uplim_{n\rightarrow\infty} \sum_{i=0}^{n} \frac{{\rm cmmse}(U_i)}{2(n+1)}   \leq \bar{I}(E; X_0) \leq \uplim_{n\rightarrow\infty}  \sum_{i=0}^{n} \frac{{\rm pmmse}(U_i)}{2(n+1)}.
		\end{equation}		
	\end{proposition}
	\begin{proof}[\rm\bf Proof]
		The MMSE-based bounds and the inequalities in \eqref{prop314_eq1} can be obtained by applying \hyperref[thm34]{Theorem 3.4} to control channel \eqref{nonlinear_ctrl_channel}. However, since both the plant and controller are nonlinear, we can no longer resort to the preceding linear Kalman filters when calculating the estimates and estimation errors of control input $U_i$ in \eqref{prop314_eq1}. Assume the control input signal $U_i$ in \eqref{nonlinear_ctrl} and \eqref{nonlinear_ctrl_channel} is observable. Otherwise, see \hyperref[sec44]{Section IV-D} for the alternative estimation and calculation scheme. In order to compute the prior estimate $\hat{U}_i^{-} := \mathbb{E}[U_i(X_i) | E_0^{i-1}]$ in ${\rm pmmse}(U_i)$ and the posterior estimate $\hat{U}_i := \mathbb{E}[U_i(X_i)|  E_0^{i}]$ in ${\rm cmmse}(U_i)$, consider the following nonlinear filtering problem with correlated noise and originated from \eqref{nonlinear_dym}-\eqref{nonlinear_ctrl_channel}:
		\begin{equation}\label{prop314_eq2}
			\begin{split}
				X_{i+1} & = \bar{f}_i(X_i) + b_i(X_i) W_i\\
				E_i & = U_i(X_i) + W_i,
			\end{split}
		\end{equation}	
		\noindent where $X_i$ denotes the hidden internal states of plant $\mathcal{P}$; the error signal $E_i$ is observable, and $\bar{f}_i(X_i) := f_i(X_i) + b_i(X_i) \cdot U_i(X_i)$. To implement the nonlinear filter and estimate the prior and posterior densities from \eqref{prop314_eq2}, we decouple the correlated noise by transforming \eqref{prop314_eq2} into the following filtering problem with independent noise \cite{Chen_RSN_2011, Wang_Auto_2014}:
		\begin{equation}\label{prop314_eq3}
			\begin{split}
				X_{i+1} & = {F}_i(X_i, E_i)\\
				E_i & = U_i(X_i) + W_i,
			\end{split}
		\end{equation}
		\noindent where ${F}_i(X_i, E_i) := \bar{f_i}(X_i) + b_i(X_i) [E_i - U_i(X_i)]$. Since $(X_i, E_i)$ in \eqref{prop314_eq3} form a Markov chain, by letting $\pi_{t|s}(x) := \mathbb{P}(X_t = x| E_0^s = e_0^s)$, we can calculate the prior density $\pi_{i+1|i}(x_{i+1}) = \mathbb{P}(X_{i+1} = x_{i+1} | E_0^i = e_0^i)$ by the time update
		\begin{equation}\label{prop314_eq4}
			\pi_{i+1 | i}(x_{i+1}) = \int_{\mathcal{X}_i} p(x_{i+1}| x_i, e_0^i) \pi_{i}(x_i)  dx_i,
		\end{equation}
		\noindent where $p(x_{i+1}|x_i, e_0^i) :=\mathbb{P}[X_{i+1} = x_{i+1} | X_i = x_i, E_0^i = e_0^i]$ is the state transition probability of \eqref{prop314_eq3}, and $\mathcal{X}_i$ denotes the state space of $X_{i}$. With the measurement $e_{i+1}$, we can compute the posterior density $\pi_{i+1}(x_{i+1}) := \mathbb{P}(X_{i+1} = x_{i+1} | E_0^{i+1} = e_0^{i+1})$ from the measurement update 
		\begin{equation}\label{prop314_eq5}
			\pi_{i+1}(x_{i+1}) = \frac{\pi_{i+1|i}(x_{i+1})  p(e_{i+1}|x_{i+1})}{\int_{\mathcal{X}_{i+1}}\pi_{i+1|i}(x_{i+1})  p(e_{i+1}|x_{i+1})  dx_{i+1}},
		\end{equation}
		\noindent where $p(e_{i+1}|x_{i+1}) := \mathbb{P}(E_{i+1} = e_{i+1} | X_{i+1} = x_{i+1})$ is the conditional distribution of observation $e_{i+1}$ given the state $x_{i+1}$. By utilizing the density functions \eqref{prop314_eq4} and \eqref{prop314_eq5}, and control mapping \eqref{nonlinear_ctrl}, we can calculate the estimates of control input $\hat{U}_i^{-}$ and $\hat{U}_i$, and thus the estimation errors ${\rm cmmse}(U_i)$ and ${\rm pmmse}(U_i)$ in \eqref{prop314_eq1}.
	\end{proof}

	It is worth noting that the density functions \eqref{prop314_eq4} and \eqref{prop314_eq5} can be approximated by the sampling methods, such as MCMC, which are arbitrarily accurate but tend to be slow and computationally intensive. In practice, by postulating some proper assumptions, e.g., Gaussian distribution assumption in \cite{Arasaratnam_PIEEE_2007}, we can simplify and expedite this approximation process with the aid of some sub-optimal filters, such as the extended Kalman filter and particle filter \cite{Chen_Statistics_2003}. Meanwhile, due to the variety and complexity of nonlinear models in \eqref{nonlinear_dym}-\eqref{nonlinear_ctrl_channel}, the control trade-off property of $\bar{I}(E;X_0)$ in \hyperref[prop314]{Proposition~3.14} is not as explicit as the linear scenarios discussed in the preceding subsections. Nevertheless, since the I-MMSE relationship and inequalities \eqref{thm34_eq1} still hold for nonlinear control systems, we can interpret the total information rate $\bar{I}(E; X_0)$ in \eqref{nonlinear_ctrl_channel} as i) a performance limit related to the average control quadratic cost as in the LTV control systems, and ii) a data transmission limit (or the plant's instability rate) of the nonlinear control system and feedback channel as we interpreted in \hyperref[sec31]{Section III-A}.

	\section{Discrete-Time Filtering Limits}\label{sec4}
	In this section, we first model the general discrete-time filtering system into an additive white Gaussian channel without feedback. With the aid of discrete-time I-MMSE relationship and optimal estimation theory, total information rate and its sandwich bounds are used to quantify the performance limits of different filtering systems.

	\subsection{General Filtering Systems and Limits}
	Consider the general discrete-time filtering system depicted in Fig.~\ref{fig3}:
	\begin{figure}[H]
		\centerline{\includegraphics[width=0.85\columnwidth]{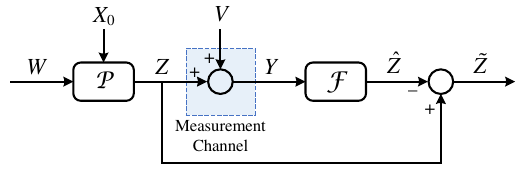}}
		\caption{Configuration of a general filtering system.}\label{fig3}
	\end{figure}	
	\noindent where $\mathcal{P}$ and $\mathcal{F}$ denote the plant model and filtering process, respectively; $(W, X_0)$ can be treated as the transmitted message of measurement channel with $W$ being a white Gaussian process noise, and $X_0$ being the initial states of $\mathcal{P}$; $V$ is a white Gaussian measurement noise; $W$, $V$, and $X_0$ are mutually independent; $Z$ stands for the noise-free output to estimate; $Y$ is the measured output; $\hat{Z}$ represents the estimate of $Z$, and the estimation error $\tilde{Z} := Z - \hat{Z}$. The open- or closed-loop plant $\mathcal{P}$ in \hyperref[fig3]{Fig.~3} is described by the following nonlinear model
	\begin{equation}\label{Filter_Plant}
		\begin{split}
			X_{i+1} & = f_i(X_i) + b_i(X_i) W_i \\
			Z_i & = h_i(X_i),
		\end{split}
	\end{equation}
	\noindent where $X_i$ and $Z_i$ are the internal states and noise-free output of $\mathcal{P}$ at time $i$; process noise $W_i \sim \mathcal{N}(0, \varepsilon^2 I)$, and $f_i(\cdot)$, $b_i(\cdot)$ and $h_i(\cdot)$ are measurable functions. Since the noise-free output $Z_i$ is a function of message $(W_0^{i-1}, X_0)$ by \eqref{Filter_Plant}, the measurement channel in \hyperref[fig3]{Fig.~3} is governed by
	\begin{equation}\label{measure_channel}
		Y_i = Z_i(W_0^{i-1}, X_0) + V_i,
	\end{equation}
	\noindent where the channel input is $Z_i(W_0^{i-1}, X_0) = Z_i = h_i(X_i)$, and the measurement noise is $V_i \sim \mathcal{N}(0, I)$. By combining the first equation in \eqref{Filter_Plant} with \eqref{measure_channel}, we can then develop a discrete-time filtering problem, in which $X_i$ is the vector of hidden states to estimate, and $Y_i$ denotes the observable output. The following assumption is then imposed on this filtering problem. 	
		
	\vspace{0.3em}
	
	\noindent (A4) \label{ass4} The noise-free output $Z_i$ is of bounded power, i.e., $\mathbb{E}[Z_i^\top Z_i] < \infty$, $\forall i$. 
	
	\begin{remark}
		\hyperref[ass4]{(A4)} is a prerequisite for applying the I-MMSE relationship or \hyperref[thm24]{Theorem~2.4} to the measurement channel \eqref{measure_channel}, since it guarantees that $Z_i$ belongs to the reproducing kernel Hilbert space induced by $V_i$, which is a necessary (and sufficient, when $Z_i$ is Gaussian) condition for the boundedness and existence of total information $I(Y_0^n; W_0^{n-1}, X_0)$ and input-output mutual information $I(Y_0^n; Z_0^n)$ in \eqref{measure_channel}. Apart from that, no restriction is imposed on the distribution of $Z_i$, i.e., $Z_i$ can be non-Gaussian.
	\end{remark}

	By applying \hyperref[thm24]{Theorem 2.4} to the filtering problem comprising of \eqref{Filter_Plant} and \eqref{measure_channel}, the following theorem shows that the total information rate $\bar{I}(Y; W, X_0)$ characterizes some performance limits in the filtering problem of estimating $Z_0^n$. 
	\begin{theorem}\label{thm42}
		For the discrete-time filtering systems subject to \eqref{Filter_Plant}, \eqref{measure_channel} and \hyperref[fig3]{Fig.~3}, the total information $I(Y_0^n;  W_0^{n-1}, X_0)$ in \eqref{measure_channel} is bounded between 
		\begin{equation}\label{thm42_eq1}
			\frac{1}{2} \sum_{i=0}^{n} {\rm cmmse}(Z_i) \leq I(Y_0^n; W_0^{n-1}, X_0) \leq \frac{1}{2} \sum_{i=0}^{n} {\rm pmmse}(Z_i),
		\end{equation}
		and the total information rate $\bar{I}(Y; W, X_0)$ satisfies\footnote{In particular, when $(Y_0^n, W_0^n)$ are stationary Gaussian or jointly stationary, by \eqref{thm42_eq1} and \hyperref[def22]{Definition 2.2}, we have $\lim_{n\rightarrow\infty}  \sum_{i=0}^{n}{\rm cmmse}(Z_i) / [2(n+1)] \leq \bar{I}(Y; W,  X_0)   \leq \lim_{n\rightarrow\infty} \sum_{i=0}^{n} {\rm pmmse}(Z_i) / [2(n+1)]$.}
		\begin{equation}\label{thm42_eq2}
			\uplim_{n\rightarrow\infty} \sum_{i=0}^{n} \frac{{\rm cmmse}(Z_i)}{2(n+1)}  \leq  \bar{I}(Y; W, X_0) \leq \uplim_{n\rightarrow\infty} \sum_{i=0}^{n} \frac{{\rm pmmse}(Z_i)}{2(n+1)}.
		\end{equation}
	\end{theorem}
	\begin{proof}[\rm\bf Proof]
		Inequalities \eqref{thm42_eq1} are obtained by applying the I-MMSE relationship in \hyperref[thm24]{Theorem~2.4} to the measurement channel \eqref{measure_channel}, and inequalities \eqref{thm42_eq2} are derived by \eqref{thm42_eq1} and the definition of mutual information rate in \hyperref[def22]{Definition~2.2}. 
	\end{proof}	    	
	\noindent \hyperref[thm42]{Theorem~4.2} shows that for the general discrete-time filtering systems, total information rate serves as a lower bound for the time-averaged prediction MMSE and an upper bound for the time-averaged causal MMSE. By resorting to the fundamental properties of differential entropy and mutual information, we have the following equality constraints and decomposition for $I(Y_0^n; W_0^{n-1}, X_0)$.  
	\begin{proposition}\label{prop43}
		For the discrete-time filtering systems subject to \eqref{Filter_Plant}, \eqref{measure_channel} and \hyperref[fig3]{Fig.~3}, the total information $I(Y_0^n; W_0^n, X_0)$ satisfies 
		\begin{equation}\label{prop43_eq1}
			I(Y_0^n; W_0^{n-1}, X_0) = I(Y_0^n; X_0) + I(Y_0^n; W_0^{n-1}|X_0),
		\end{equation}
		and
		\begin{equation}\label{prop43_eq2}
			I(Y_0^n; W_0^{n-1}, X_0) = h(Y_0^n) - h(V_0^n). 
		\end{equation}
	\end{proposition}
	\begin{proof}[\rm\bf Proof]
		\eqref{prop43_eq1} is obtained by applying the chain rule of mutual information to $I(Y_0^n; W_0^{n-1}, X_0)$. \eqref{prop43_eq2} can be verified by
		\begin{align*}
			I(Y_0^n; W_0^{n-1}, X_0) & = h(Y_0^n) - h(Y_0^n|W_0^{n-1}, X_0) \allowdisplaybreaks\\
			& = h(Y_0^n) - \sum_{i=0}^n h(V_i|Y_0^{i-1}, W_0^{i-1}, X_0) \allowdisplaybreaks\\
			& = h(Y_0^n) - \sum_{i=0}^n h(V_i|V_0^{i-1}) ,\allowdisplaybreaks
		\end{align*}
		which follows from the same arguments in proving \eqref{prop33_eq2}.
	\end{proof}
	\noindent The discussion of total information (rate) serving as a filtering trade-off metric in \hyperref[thm42]{Theorem 4.2} and \hyperref[prop43]{Proposition 4.3} is symmetric and can be implied as the filtering and non-feedback counterpart of the control trade-off interpretations following \hyperref[thm32]{Theorem 3.2} and \hyperref[prop33]{Proposition 3.3}.

	By incorporating the bounds and constraints in \hyperref[thm42]{Theorem 4.2} and \hyperref[prop43]{Proposition 4.3} together, a comprehensive description of the total information rate $\bar{I}(Y; W, X_0)$ as a filtering limit metric can be obtained. Before showing this result, we define the ``filtering capacity" $\mathcal{C}_f := \sup_{(W, X_0, Z)} \bar{I}(Y; W, X_0)$ as the supremum of total information rate over all admissible pairs of the ``message" $(W, X_0)$ and the noise-free output $Z$. When the power of noise-free output is limited by $\mathbb{E}[Z_i^\top Z^{}_i] \leq \rho_i$, $\forall i$, the channel capacity $\mathcal{C}:= \sup_{\mathbb{E}[Z_i^\top Z_i] \leq \rho_i} \bar{I}(Y; Z)$ is defined as the supremum of the input-output mutual information rate. The following theorem then gives a full description on the total information rate $\bar{I}(Y; W, X_0)$ in the general filtering systems.   
	\begin{theorem}\label{thm44}
		For the discrete-time filtering systems subject to \eqref{Filter_Plant}, \eqref{measure_channel} and \hyperref[fig3]{Fig.~3}, the total information rate $\bar{I}(Y; W, X_0)$ in \eqref{measure_channel} is bounded between 
		\begin{equation}\label{thm44_eq1}
			{\rm FLB} \leq \bar{I}(Y; W, X_0) \leq \mathcal{C}_f \leq \uplim_{n\rightarrow\infty} \sum_{i=0}^{n} \frac{{\rm pmmse}(Z_i)}{2(n+1)},
		\end{equation}
		where the filtering lower bound is ${\rm FLB} = \max\{\uplim_{n\rightarrow\infty}  [2(n \break +1)]^{-1} \sum_{i=0}^{n}{\rm cmmse}(Z_i),  \bar{I}(Y; X_0) \}$. When $\mathbb{E}[Z_i^\top Z_i] \leq \rho_i$, $\forall i$, and $\Sigma_{Z_i}$ denotes the covariance of $Z_i$, the filtering capacity is given by $\mathcal{C}_f = \mathcal{C} = \overline{\lim}_{n\rightarrow\infty}\sup_{\mathbb{E}[Z_i^\top Z_i] \leq \rho_i}[2(n+1)]^{-1}\cdot \break \sum_{i=0}^{n} \log(|\Sigma_{Z_i}  + I|)$. When the process noise $W_i$ vanishes, i.e., $\varepsilon \rightarrow 0$ in \eqref{Filter_Plant}, the total information rate $\bar{I}(Y; X_0) = \lim_{\varepsilon \rightarrow 0} \bar{I}(Y; W, X_0)$ in \eqref{measure_channel} satisfies
		\begin{equation}\label{thm44_eq2}
			\scaleto{\lim_{\varepsilon \rightarrow 0}   \uplim_{n\rightarrow\infty}     {\textstyle \sum_{i=0}^{n} \limits } \frac{{\rm cmmse}(Z_i)}{2(n+1)}  }{22pt}    \leq  \bar{I}(Y; X_0) 
			\leq   \scaleto{ \lim_{\varepsilon \rightarrow 0}  \uplim_{n\rightarrow\infty}     {\textstyle \sum_{i=0}^{n} \limits} \frac{{\rm pmmse}(Z_i)}{2(n+1)}  }{22pt}  .
		\end{equation}
	\end{theorem}
	\begin{proof}[\rm\bf Proof] 
		Since \hyperref[thm44]{Theorem 4.4} can be regarded as the filtering and non-feedback counterpart of \hyperref[thm34]{Theorem 3.4}, its proof can be inferred from the proof of \hyperref[thm34]{Theorem 3.4} by symmetry, and will not be presented in detail. One only need to notice that when the process noise $W_i \sim \mathcal{N}(0, \varepsilon^2 I)$ vanishes, i.e., $\varepsilon\rightarrow 0$ in \eqref{Filter_Plant}, the total information rate $\bar{I}(Y; W, X_0)$ in \eqref{measure_channel} satisfies
		\begin{equation}\label{thm44_eq5}
			\bar{I}(Y; W, X_0) \geq \lim_{\varepsilon \rightarrow 0} \bar{I}(Y; W, X_0) = \bar{I}(Y; X_0),
		\end{equation}
		\noindent which follows from \eqref{prop43_eq1} and the fact that when $W_0^n \rightarrow 0$ and becomes deterministic, $\lim_{\varepsilon \rightarrow 0} I(Y_0^n; W_0^n|X_0)  = 0$. Interested readers are referred to \hyperref[appD]{Appendix D} or \cite[Thm. 3.22]{Wan_2024} for the detailed proof of \hyperref[thm44]{Theorem 4.4}.
	\end{proof}   
	\noindent In addition to the trade-off properties characterized in \hyperref[thm42]{Theorem 4.2} and \hyperref[prop43]{Proposition 4.3}, \hyperref[thm44]{Theorem 4.4} and inequalities \eqref{thm44_eq2} reveal that for the general discrete-time filtering systems under vanishing process noise, the total information rate defines a lower bound for the {\it lowest achievable} prediction MMSE of noise-free output irrespective of the design of filtering process $\mathcal{F}$. In other words, no matter how well the filter $\mathcal{F}$ in \hyperref[fig3]{Fig.~3} is designed, the time-averaged prediction MMSE of $Z_i$ cannot be smaller than twice $\bar{I}(Y; X_0)$ in practice (or in the presence of process noise). Meanwhile, we will later show that when there exists an analytic or explicit relationship between the prediction and causal MMSEs, $\bar{I}(Y; X_0)$ can also be used to define a lower bound for the lowest achievable causal MMSE. In the following subsections, with the aid of \hyperref[thm44]{Theorem~4.4}, $\bar{I}(Y; X_0)$ is investigated as a filtering trade-off metric for more specific filtering systems, in which the closed-form expression of $\bar{I}(Y; X_0)$ can be derived.

\subsection{Filtering Limits in LTI Systems}
	We now narrow down our scope to the filtering problem depicted by the following LTI plant $\mathcal{P}$ and measurement model
	\begin{equation}\label{LTI_Filtering}
		\begin{split}
			X_{i+1} & = AX_i + BW_i\\
			Y_i & = Z_i + V_i = HX_i+ V_i ,
		\end{split}
	\end{equation}
	\noindent where $A$, $B$, and $H$ are the time-invariant matrices of proper dimensions, and process noise $W_i \sim \mathcal{N}(0, \varepsilon^2 I)$ and measurement noise $V_i \sim \mathcal{N}(0, I)$ are mutually independent. Since $W_i$ and $V_i$ in \eqref{LTI_Filtering} are uncorrelated and zero-mean white Gaussian, the optimal mean-square (MS) estimates and the minimum MS estimation errors of filtering problem \eqref{LTI_Filtering} can be calculated by the LTI Kalman filter \cite{Simon_2006}. By applying \hyperref[thm44]{Theorem 4.4} and the LTI Kalman filter to \eqref{LTI_Filtering} under vanishing process noise, \hyperref[prop45]{Proposition 4.5} then interprets the filtering trade-off property of $\bar{I}(Y; X_0)$. 
	\begin{proposition}\label{prop45}
		For the filtering system subject to the stabilizable and detectable LTI plant \eqref{LTI_Filtering} and vanishing process noise ($\varepsilon \rightarrow 0$) in \hyperref[fig3]{Fig. 3}, the total information rate $\bar{I}(Y; X_0)$ in \eqref{LTI_Filtering} satisfies
		\begin{equation}\label{prop45_eq1}
			\scaleto{\lim_{\varepsilon \rightarrow 0}   \uplim_{n\rightarrow\infty}     {\textstyle \sum_{i=0}^{n} \limits } \frac{{\rm cmmse}(Z_i)}{2(n+1)}  }{22pt}  \leq \bar{I}(Y; X_0)   \leq \scaleto{\lim_{\varepsilon \rightarrow 0}   \uplim_{n\rightarrow\infty}     {\textstyle \sum_{i=0}^{n} \limits } \frac{{\rm pmmse}(Z_i)}{2(n+1)}  }{22pt} 
		\end{equation}
		\noindent where $\bar{I}(Y; X_0) = \sum_j \log |\lambda_j^{+}(A)|$, with $\lambda_j^{+}(A)$ denoting the eigenvalues of matrix $A$ with positive real parts, i.e., open-loop unstable poles of $\mathcal{P}$ in \hyperref[fig3]{Fig. 3}. 
	\end{proposition}
	\begin{proof}[\rm\bf Proof]
		Since \hyperref[prop45]{Proposition 4.5} can be proved by replacing the feedback elements in the proof of \hyperref[prop36]{Proposition 3.6} with the non-feedback ones in \eqref{LTI_Filtering}, we skip this proof. Interested readers are referred to \hyperref[appE]{Appendix E} or \cite[Proposition 3.23]{Wan_2024} for details.
	\end{proof}
	\noindent \hyperref[prop45]{Proposition 4.5} shows that information rate $\bar{I}(Y; X_0)$ serves as a time-domain trade-off metric in the filtering systems subject to the LTI plant \eqref{LTI_Filtering}, i.e., regardless of the design of the filtering mapping $\mathcal{F}$ in \hyperref[fig3]{Fig. 3}, in practice or in the presence of process noise, if the prediction MMSE of $Z_i$ is smaller than $2\bar{I}(Y; X_0)$ in a time interval, the estimation error outside this interval must be larger, and vice versa. Since $\bar{I}(Y;  X_0) = \sum_j \log |\lambda_j^{+}(A)|$, the lowest achievable prediction MMSE of $Z_i$ is determined by the unstable dynamics in \eqref{LTI_Filtering}. If $W_i$ and $V_i$ in \eqref{LTI_Filtering} are correlated or non-Gaussian, $\sum_j \log |\lambda_j^{+}(A)|$ then only defines the filtering limit of all linear filters. By utilizing the analytic relation between the prediction and causal MMSEs, a lower bound for the lowest achievable causal MMSE of $Z_i$ can also be defined by $\bar{I}(Y; X_0)$.\footnote{\label{footnote5}Pre- and post-multiplying \eqref{LTI_fl_correction} by $H$ and $H^\top$, and using the Woodbury matrix identity, we have $HP_iH^\top = HP_i^{-}H^\top (I + HP_i^{-}H^\top)^{-1}$. Taking ${\rm tr}(\cdot)$ on both sides of this equation, and using \hyperref[lem46]{Lemma~4.6} in the limits of $\varepsilon \rightarrow 0$ and $n\rightarrow\infty$, we have $\lim_{\varepsilon \rightarrow 0}\overline{\lim}_{n\rightarrow \infty} \sum_{i=0}^{n}{\rm cmmse}(Z_i)/(n+1) = 2\sum_{k} \eta_k / (1 + \eta_k)$, where $\{\eta_k\}$ denotes the eigenvalues of $H_uP_u^{-}H_u^\top$ in \eqref{prop45_eq2}. When $\{\eta_k\}$ are known, we can directly calculate the lowest achievable casual MMSE by $2\sum_{k} \eta_k / (1 + \eta_k)$. When $\bar{I}(Y; X_0)$ is known, we can estimate a lower bound for the lowest achievable causal MMSE by solving the minimization problem $\min \sum_k \eta_k / (1+\eta_k)$, s.t. $\sum_k \eta_k \geq 2\bar{I}(Y; X_0)$, and $\eta_k > 0, \forall k$.}

	\subsection{Filtering Limits in LTV Systems}
	We then consider the filtering systems subject to the following LTV plant and measurement model
	\begin{equation}\label{LTV_Filtering}
		\begin{split}
			X_{i+1} & = A_i X_i + B_i W_i\\
			Y_{i} & = Z_i + V_i = H_i X_i + V_i,
		\end{split}
	\end{equation}
	\noindent where $A_i$, $B_i$ and $H_i$, satisfying \hyperref[ass2]{(A2)} and \hyperref[ass3]{(A3)}, are time-varying matrices of proper dimensions; process noise $W_i \sim \mathcal{N}(0, \varepsilon^2 I)$ and measurement noise $V_i \sim \mathcal{N}(0, I)$ are mutually independent. Since $W_i$ and $V_i$ in \eqref{LTV_Filtering} are uncorrelated and zero-mean white Gaussian, we can use the LTV Kalman filter to infer the optimal MS estimates and minimum MS estimation errors of filtering problem \eqref{LTV_Filtering}. The following lemma shows that when the process noise $W_i$ in \eqref{LTV_Filtering} vanishes, the state prior and posterior error covariance matrices of filtering problem \eqref{LTV_Filtering} are determined by the antistable dynamics. 
	\begin{lemma}\label{lem46}
		For the uniformly completely reconstructible and stabilizable LTV system \eqref{LTV_Filtering} subject to vanishing process noise, i.e., $\varepsilon\rightarrow 0$, the asymptotic prior error covariance of its modal-decomposed system satisfies
		\begin{equation}\label{lem46_eq1}
			\lim_{\varepsilon \rightarrow 0} \lim_{i'\rightarrow\infty} \mathbb{E}\left\{ \left[ \begin{matrix}
				\tilde{X}^{-}_s(i')\\
				\tilde{X}^{-}_u(i')
			\end{matrix}   \right]   
			\left[\begin{matrix}
				\tilde{X}^{-}_s(i')\\
				\tilde{X}^{-}_u(i')
			\end{matrix}\right]^\top
			\right\} = \left[ \begin{matrix}
				0 & 0\\
				0 & \bar{P}_{u}^{-}(i)
			\end{matrix}\right],
		\end{equation}
		\noindent where $\tilde{X}_s^{-}(i') := X^{}_s(i') - \hat{X}^{-}_s(i')$ and $\tilde{X}_u^{-}(i') := X_u(i') - \hat{X}^{-}_u(i')$ respectively denote the prior errors of stable and antistable modes, and $\bar{P}_u^{-}(i) := \lim_{\varepsilon \rightarrow 0} \lim_{i'\rightarrow\infty}  \mathbb{E}[\tilde{X}_u^{-}(i') \cdot \break \tilde{X}_u^{-}(i')^\top ]$ is the asymptotic prior error covariance of antistable mode that satisfies $\bar{P}_{u}^{-}(i+1) = A^{}_u(i)\bar{P}^{-}_u(i) [H^{}_u(i) \cdot\bar{P}^{-}_u(i) H^\top_u(i) + I]^{-1}  A_u^\top(i)$, with $A_u$ and $H_u$ being the antistable parts of $A$ and $H$. The asymptotic posterior error covariance of the modal-decomposed system satisfies
		\begin{equation}\label{lem46_eq2}
			\lim_{\varepsilon \rightarrow 0} \lim_{i'\rightarrow\infty} \mathbb{E}\left\{ \left[ \begin{matrix}
				\tilde{X}_s(i')\\
				\tilde{X}_u(i')
			\end{matrix}   \right]   
			\left[\begin{matrix}
				\tilde{X}_s(i')\\
				\tilde{X}_u(i')
			\end{matrix}\right]^\top
			\right\} = \left[ \begin{matrix}
				0 & 0\\
				0 & \bar{P}_{u}(i)
			\end{matrix}\right],
		\end{equation}
		\noindent where $\tilde{X}_s(i') := X_s(i') - \hat{X}_s(i')$ and $\tilde{X}_u(i') := X_u(i') - \hat{X}_u(i')$ are respectively the posterior errors of stable and antistable modes, and $\bar{P}_u(i) := \lim_{\varepsilon \rightarrow 0} \lim_{i'\rightarrow\infty}  \mathbb{E}[\tilde{X}_u(i') \cdot \tilde{X}_u(i')^\top ]$ is the asymptotic posterior error covariance of anti-stable mode and satisfies $\bar{P}_u(i) = \bar{P}_u^{-}(i) - \bar{P}_u^{-}(i) H_u^\top(i) \cdot [H^{}_u(i) P_u^-(i)H_u^\top(i) + I]^{-1}H_u(i)P_u^{-}(i)$. 
	\end{lemma}  
	\begin{proof}[\rm\bf Proof] 
		See \hyperref[appF]{Appendix F} or \cite[Lemma 3.24]{Wan_2024} for the proof. 
	\end{proof}    
	\noindent \hyperref[lem46]{Lemma 4.6} also holds for the LTI filtering systems \eqref{LTI_Filtering} and has been previously invoked in the proofs of \hyperref[prop36]{Propositions 3.6}, \hyperref[prop311]{3.11} and \hyperref[prop45]{4.5}. By applying \hyperref[thm44]{Theorem 4.4} and \hyperref[lem46]{Lemma 4.6} to the LTV filtering problem and measurement channel \eqref{LTV_Filtering}, we then have the following result on $\bar{I}(Y; X_0)$.     
	\begin{proposition}\label{prop47}
		For the discrete-time filtering systems depicted in \hyperref[fig3]{Fig. 3} and uniformly completely stabilizable and reconstructible LTV plant \eqref{LTV_Filtering} under vanishing process noise, the total information rate $\bar{I}(Y; X_0)$ in \eqref{LTV_Filtering} satisfies
		\begin{align}\label{prop47_eq1}
			{\rm FLB}_{\rm LTV} \leq \bar{I}(Y; X_0)  	 \leq \lim_{\varepsilon \rightarrow 0} \uplim_{n\rightarrow\infty}  \sum_{i=0}^{n} \frac{{\rm pmmse} (Z_i)}{2(n+1)},
		\end{align}
		\noindent where $\bar{I}(Y; X_0) = \overline{\lim}_{n\rightarrow\infty}(n+1)^{-1} \log\det[\varPhi_{A_u}(n+1, 0)]  = \mathscr{B}$, and the filtering lower bound is ${\rm FLB}_{\rm LTV} =  \max\{     \lim_{\varepsilon  \rightarrow 0} \break \uplim_{n\rightarrow\infty} [2(n+1)]^{-1}  \sum_{i=0}^{n}  {{\rm cmmse} (Z_i)}, \sum_{j=1}^{l}  m_j    \log \underline{\kappa}_j  \}$. 
	\end{proposition}
	\begin{proof}[\rm\bf Proof]
		Since \hyperref[prop47]{Proposition 4.7} can be virtually proved by replacing the feedback elements in the proof of \hyperref[prop311]{Proposition 3.11} with the non-feedback counterparts in \eqref{LTV_Filtering}, we skip this proof. Interested readers are referred to \hyperref[appG]{Appendix G} or \cite[Prop. 3.25]{Wan_2024} for the detailed proof.
	\end{proof}    
	\noindent \hyperref[prop47]{Proposition 4.7} shows that for the filtering problem subject to the LTV plant and measurement channel in \eqref{LTV_Filtering} and  \hyperref[fig3]{Fig. 3}, total information rate $\bar{I}(Y; X_0)$ characterizes a lower bound for the lowest achievable prediction MMSE of $Z_i$ irrespective of the design of the filtering process $\mathcal{F}$ in \hyperref[fig3]{Fig. 3}. Specifically, in the presence of process noise $W_i$ or $\varepsilon^2 > 0$ in \eqref{LTV_Filtering}, no matter how well the filtering process $\mathcal{F}$ is designed, the time-averaged prediction MMSE of $Z_i$ must be larger than $2\bar{I}(Y; X_0)$, which is quantified by the antistable dynamics of \eqref{LTV_Filtering}. When the process and measurement noises are correlated or non-Gaussian, $\bar{I}(Y; X_0)$ only defines the lower bounds for all linear filters. Similar to the discussion$^{\ref{footnote5}}$ beneath \hyperref[prop45]{Proposition 4.5}, with the analytic relation between the prediction and causal MMSEs, $\bar{I}(Y; X_0)$ can also be used to define a lower bound for the lowest achievable causal MMSE of $Z_i$.

	\subsection{Filtering Limits in Nonlinear Systems}\label{sec44}
	Lastly, consider the nonlinear filtering system subject to the nonlinear plant \eqref{Filter_Plant} and measurement model \eqref{measure_channel}. According to \hyperref[thm44]{Theorem 4.4}, in the presence of process noise or $\varepsilon^2 > 0$ in \eqref{Filter_Plant}, the information rate $\bar{I}(Y; W, X_0)$ characterizes a lower bound for the time-averaged or steady-state prediction MMSE of $Z_i$ in \eqref{thm44_eq1}, the value of which, however, is sensitive to the noise level $\varepsilon$. When the process noise vanishes or $\varepsilon \rightarrow 0$ in \eqref{Filter_Plant}, the information rate $\bar{I}(Y; X_0) =  \lim_{\varepsilon\rightarrow0}\bar{I}(Y; W, X_0)$, obeying \eqref{thm44_eq2}, then defines a lower bound for the lowest achievable prediction MMSE regardless of the design of filtering process $\mathcal{F}$.

	Hence, it only remains to solve how to calculate the MMSE-based bounds or estimate the information rates, $\bar{I}(Y; W, X_0)$ and $\bar{I}(Y; X_0)$, in \hyperref[thm44]{Theorem 4.4}, especially, when the noise-free output or channel input $Z_i$ is unobservable, compared with the observable control input signal considered in \hyperref[sec34]{Section III-D}. When $Y_i$ and $Z_i$ in \eqref{Filter_Plant} and \eqref{measure_channel} are Gaussian or normal, using the theorem on normal correlation \cite[Ch. 13]{Liptser_2001}, we can compute the MMSEs of $Z_i$ by
	\begin{align}\label{normal_correlation}
		\hspace{-5pt}\mathbb{E}[(Z_i - \hat{Z}_i)^\top (Z_i - \hat{Z}_i)] & = \mathbb{E}[(Z_i - \hat{Z}_i)^\top (Z_i - \hat{Z}_i) | Y_0^i] \allowdisplaybreaks  \\
		& \hspace{-8pt} = \mathbb{E}[Z_i^\top Z_i | Y_0^i] - \mathbb{E}[Z_i|Y_0^i]^\top \mathbb{E}[Z_i|Y_0^i], \nonumber
	\end{align}
	\noindent which also holds when we replace the posterior estimate $\hat{Z}_i$ and $Y_0^i$ by the prior estimate $\hat{Z}^{-}_i$ and $Y_0^{i-1}$. To calculate the expectations in \eqref{normal_correlation}, we need to figure out the prior density $\pi_{i|i-1}(x_{i}) := \mathbb{P}(X_{i} = x_{i} | Y_0^{i-1} = y_0^{i-1})$ and posterior density $\pi_{i}(x_{i}) := \mathbb{P}(X_{i} = x_{i} | Y_0^{i} = y_0^{i})$. Since $X_i$ in \eqref{Filter_Plant} is Markov, we can calculate the prior density by the time update step   	
	\begin{equation}\label{prior_density}
		\pi_{i |i-1}(x_{i}) = \int_{\mathcal{X}_{i-1}}p(x_{i}|x_{i-1}) \pi_{i-1}(x_{i-1}) dx_{i-1},
	\end{equation}
	\noindent where $p(x_{i}|x_{i-1}):= \mathbb{P}(X_{i} = x_{i} | X_{i-1} = x_{i-1})$ is the state transition probability of \eqref{Filter_Plant}, and $\mathcal{X}_{i-1}$ denotes the state space of $X_{i-1}$. With the output measurement $y_{i+1}$, we can calculate the posterior density by the measurement update step
	\begin{equation}\label{posterior_density}
		\pi_{i}(x_{i}) = \frac{\pi_{i|i-1}(x_{i}) p(y_{i}| x_{i})}{\int_{\mathcal{X}_i}\pi_{i|i-1}(x_{i}) p(y_{i}|x_{i})dx_{i}}, 
	\end{equation}
	\noindent where $p(y_{i}|x_{i})$ denotes the conditional distribution of observation $y_i$ given state $x_i$. With the prior and posterior densities in \eqref{prior_density} and \eqref{posterior_density}, we can evaluate ${\rm pmmse}(Z_i)$ and ${\rm cmmse}(Z_i)$, and hence sandwich estimate $\bar{I}(Y; W, X_0)$ and $\bar{I}(Y; X_0)$ by \eqref{normal_correlation} and \hyperref[thm44]{Theorem 4.4}.

	Similar to the challenges of nonlinear control scenario in \hyperref[sec34]{Section III-D}, accurate computation of densities in \eqref{prior_density} and \eqref{posterior_density} is slow and computationally intensive. By adopting some proper assumptions, e.g., Gaussian distribution assumption, and replacing the optimal estimator with the sub-optimal filters, e.g., extended Kalman filter, we can expedite the process of estimating $\pi_{i |i-1}(x_{i})$, $\pi_{i}(x_{i})$ and the associated MMSEs. Due to the complexity of nonlinear dynamical models, information rates, $\bar{I}(Y; W, X_0)$ and $\bar{I}(Y; X_0)$, are not directly tied to the (antistable dynamics of) nonlinear plant \eqref{Filter_Plant}. However, the above analysis and \hyperref[thm44]{Theorem~4.4} indicate that for nonlinear filtering systems, $\bar{I}(Y; W, X_0)$ or $\bar{I}(Y; X_0)$ still quantifies a lower bound for the (lowest achievable) prediction MMSE of  $Z_i$. Similar to the linear filtering cases, when the relationship between the prediction and causal MMSEs is analytic, it is possible to define a lower bound for the (lowest achievable) causal MMSE of $Z_i$ by using $\bar{I}(Y; W, X_0)$ or $\bar{I}(Y; X_0)$.

	\section{Conclusion and Discussion}\label{sec5}
	In this paper, we investigated the total information rate as a control and filtering trade-off metric by using the I-MMSE relationships. For the first time, we derived and extended the I-MMSE relationship into the discrete-time additive white Gaussian channels with feedback. For the control systems, we showed that the total information rate is sandwiched by the time-averaged causal and prediction MMSEs of the control input, and serves as an information-theoretic interpretation of the Bode-type integrals and average risk-sensitive cost function. Specifically, when the plant and controller are respectively LTI, LTV, and nonlinear, based on the optimal filtering theory, we calculated or estimated the values of total information rate and its MMSE-based bounds. For the filtering systems, we showed that the total information rate is bounded between the time-averaged causal and prediction MMSEs of the noise-free output, and characterizes a lower bound for the lowest achievable prediction and causal MMSEs of noise-free output. When the plant generating output is respectively LTI, LTV, and nonlinear, by resorting to the optimal filtering theory, we calculated or estimated the values of total information rate and its MMSE-based bounds. Some promising buildups of this paper include i) using the I-MMSE relationship proposed in \hyperref[thm24]{Theorem 2.4} to revisit the discrete-time AWGN channels, e.g., estimating channel capacity or verifying data transmission rate of coding scheme, and ii) extending the I-MMSE relationships and the trade-off analyses presented in this paper to the control and filtering channels and systems subject to colored and non-Gaussian noises, e.g. \cite[Chapter 4]{Wan_2024}.

	\bibliographystyle{IEEEtran}
	\bibliography{myref}
	
	\newpage
	
	\vskip -2.1\baselineskip plus -1fil
	
	\begin{biography}[{\includegraphics[width=1in, height=1.25in, clip, keepaspectratio]{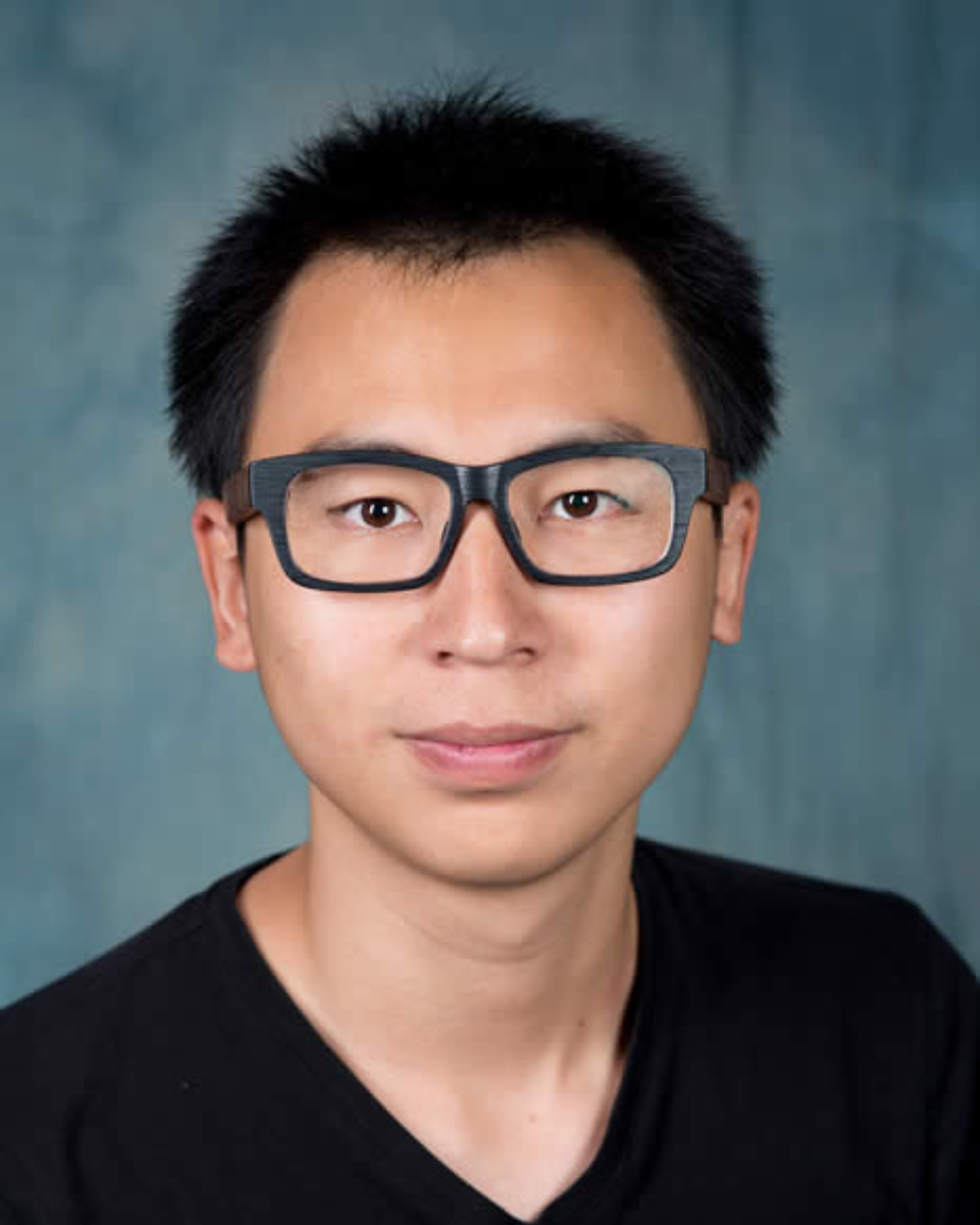}}]{Neng Wan} received the bachelor’s (Hons.) and master’s degrees in aerospace engineering from Harbin Institute of Technology, Harbin, China, in 2013 and 2015, respectively, and the master’s degree in applied and computational mathematics from the University of Minnesota Duluth, Duluth, MN, USA, in 2017. He is currently pursuing the Ph.D. degree with the Department of Mechanical Science and Engineering and Coordinated Science Laboratory, University of Illinois at Urbana-Champaign, Urbana, IL, USA. 
		
	His research interests include control theory, information theory, machine learning and their applications.
	\end{biography}

	\vskip -2.1\baselineskip plus -1fil
	
	\begin{biography}[{\includegraphics[width=1in, height=1.25in, clip, keepaspectratio]{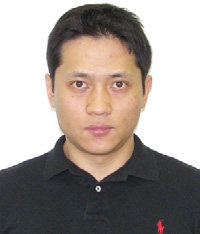}}]{Dapeng Li} is a Research Professor of mechanical and energy engineering with Southern University of Science and Technology in Shenzhen, China. He received his B.S. and M.S. degrees in Automatic Control from University of Science and Technology of China. He got his Ph.D. degree in Mechanical Engineering in 2011 from University of Illinois at Urbana-Champaign. 
		
	His research interests include the theory of robust adaptive control, information/control theory, machine learning, and their applications in HVAC and aerospace systems.	 
	\end{biography}

	\vskip -2.1\baselineskip plus -1fil
		
	\begin{biography}[{\includegraphics[width=1in, height=1.25in,clip, keepaspectratio]{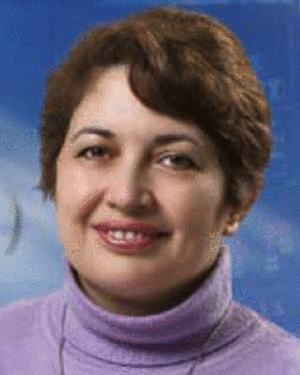}}]{Naira Hovakimyan} (Fellow, IEEE)  received the M.S. degree in applied mathematics from Yerevan State University, Yerevan, Armenia, and the Ph.D. degree in physics and mathematics from the Institute of Applied Mathematics of Russian Academy of Sciences, Moscow, Russia.
		
	She is currently a W. Grafton and Lillian B. Wilkins Professor of mechanical science and engineering with the University of Illinois at Urbana-Champaign, Urbana, IL, USA. She has co-authored 2 books and more than 450 refereed publications. She holds 11 patents.
	
	Dr. Hovakimyan was the recipient of AIAA Mechanics and Control of Flight Award in 2011, SWE Achievement Award in 2015, IEEE CSS Award for Technical Excellence in Aerospace Controls in 2017, and AIAA Pendray Aerospace Literature Award in 2019. In 2014, she received the Humboldt prize for her lifetime achievements. She is a Fellow of AIAA. Her work was featured in The New York Times, on Fox TV and CNBC.	
	\end{biography}

	\vskip -2.1\baselineskip plus -1fil
	
	\begin{biography}[{\includegraphics[width=1in, height=1.25in,clip, keepaspectratio]{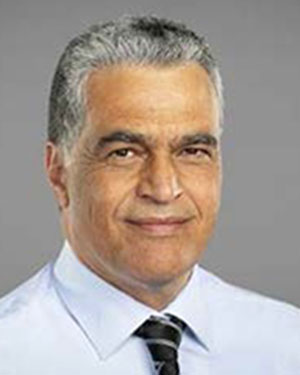}}]{Petros G. Voulgaris} (Fellow, IEEE)  received the Diploma in Mechanical Engineering from the National Technical University, Athens, Greece, in 1986, and the S.M. and Ph.D. degrees in Aeronautics and Astronautics from the Massachusetts Institute of Technology in 1988 and 1991, respectively.
		
		He is currently Chair, Founding Aerospace Program Director, and Victor LaMar Lockhart Professor in Mechanical Engineering at University of Nevada, Reno. Before joining UNR in 2020 and since 1991, he has been a faculty with the Department of Aerospace Engineering, University of Illinois at Urbana-Champaign holding also appointments with the Coordinated Science Laboratory, and the department of Electrical and Computer Engineering. His research interests are in the general area of robust and optimal control and coordination of autonomous systems.\allowdisplaybreaks
		
		Dr. Voulgaris is a recipient of several awards including the NSF Research Initiation Award, the ONR Young Investigator Award and the UIUC Xerox Award for research. He has also been a Visiting ADGAS Chair Professor, Mechanical Engineering, Petroleum Institute, Abu Dhabi, UAE and a Visiting Guangbiao Chair at Zhejiang University, China. His research has been supported by several agencies including NSF, ONR, AFOSR, NASA. He is also a Fellow of IEEE.	
	\end{biography}

\clearpage

\section*{Appendix A: Proof of Proposition~3.6}\label{appA}
A sketch and an overview of the proof are first given. The two inequalities in \eqref{prop36_eq1} can be directly obtained by applying \hyperref[thm32]{Theorem 3.2} or \ref{thm34} to the LTI control channel \eqref{LTI_Ctrl_Channel}. It then remains to verify that $\bar{I}(E; X_0) = \sum_j \log |\lambda_j^{+}(A)|$, which is an important evidence for showing that $\bar{I}(E; X_0)$ serves as an information-theoretic interpretation of the Bode's integral of discrete-time LTI control systems \cite{Sung_IJC_1989}. We can derive this equality by separately proving the statements i) $\bar{I}(E; X_0) \geq \sum_j \log |\lambda_j^{+}(A)|$, and ii) $\bar{I}(E; X_0) \leq \sum_j \log |\lambda_j^{+}(A)|$. Statement i) is well-known and has been proved by using the internally mean-square stability and the maximum entropy conditions in \cite[Lemma 4.1]{Martins_TAC_2007}. However, to the authors' knowledge, statement ii) has not been proved in any existing literature yet. Surprisingly, even without proving statement ii), equality $\bar{I}(E; X_0) = \sum_j \log |\lambda_j^{+}(A)|$ was incidentally derived by \cite[Corollary 1]{Tanaka_TAC_2018} during the minimization of the directed information (rate) $I(U_0^n \rightarrow E_0^n)  = I(E_0^n; X_0)$ subject to a LQG performance constraint. In the following, we will prove statement ii) by using tools from information and estimation theories, which along with the proof of statement i), provides a more straightforward and concise approach to derive the identity $\bar{I}(E; X_0) = \sum_j \log |\lambda_j^{+}(A)|$.

In order to calculate the estimates and estimation error (covariance) of control input $U_i = GX_i$, consider the following filtering problem originated from \eqref{LTI_Ctrl_Plant}-\eqref{LTI_Ctrl_Channel} 
\begin{equation}\label{LTI_Ctrl_Filtering}
	\begin{split}
		X_{i+1} & = (A+BG)X_i +BW_i\\
		E_i & = GX_i + W_i,
	\end{split}
\end{equation}
where $E_i$ is observable, and $X_i$ is the hidden states to estimate. By applying the LTI Kalman filter to \eqref{LTI_Ctrl_Filtering}, and letting the prior estimate be $\hat{X}_i^{-} := \mathbb{E}[X_i|E_0^{i-1}]$, the non-zero prior error covariance $P_i^{-} := \mathbb{E}[(X_i - \hat{X}_i^{-})(X_i - \hat{X}_i^{-})^\top]$ satisfies
\begin{equation}\label{LTI_Prediction}
	P_{i}^{-} = AP_{i-1}^{-}(GP^{-}_{i-1}G^\top + I)^{-1} A^\top.
\end{equation}
Define the posterior estimate $\hat{X}_i = \mathbb{E}[X_i | E_0^i]$, and the posterior error covariance $P_i := \mathbb{E}[(X_i - \hat{X}_i)(X_i - \hat{X}_i)^\top]$ obeys
\begin{equation}\label{LTI_Correction}
	P_{i} = P_{i}^{-}-P_{i}^{-}G^\top(GP_{i}^{-}G^\top + I)^{-1}GP_{i}^{-}.
\end{equation}
The evolution of error covariance matrices in \eqref{LTI_Prediction} and \eqref{LTI_Correction} is equivalent to the filtering result of the following system without control input and noise\footnote{The filtering equivalence only exists for the prior/posterior error covariance matrices of \eqref{LTI_Ctrl_Filtering} and \eqref{prop36_eq4.5}, while the prior/posterior estimates of these two systems are not identical.}
\begin{equation}\label{prop36_eq4.5}
	\begin{split}
		X_{i+1} & = AX_i, \allowdisplaybreaks\\
		E_i & = GX_i + W_i.
	\end{split}
\end{equation}
When \eqref{LTI_Ctrl_Filtering} is detectable, we can partition \eqref{LTI_Ctrl_Filtering} and \eqref{prop36_eq4.5} into stable and unstable parts by modal decomposition. By the LTI version of \hyperref[lem46]{Lemma~4.6}, the asymptotic prior error covariance $\lim_{i\rightarrow\infty}P_i^{-}$ in \eqref{LTI_Prediction} takes the form of  
\begin{equation}\label{prop36_eq5}
	\lim_{i\rightarrow\infty}P_i^{-} = \left[\begin{matrix}
		0 & 0\\
		0 & P_u^{-}
	\end{matrix}\right] ,
\end{equation}
where $P_u^{-} := \lim_{i\rightarrow\infty} \mathbb{E}[(X_u(i) - \hat{X}_u^{-}(i))(X_u(i) - \hat{X}_u^{-}(i))^{\top}]$, with unstable modes $X_u(i)$ and their prior estimates $\hat{X}_u^{-}(i)$, is the asymptotic prior error covariance of the unstable part and satisfies the steady-state ARE (algebraic Riccati equation)
\begin{equation}\label{LTI_ARE_Ctrl}
	P_{u}^{-} = A^{}_u P_u^{-} (G^{}_u P_u^{-} G_u^\top + I)^{-1} A_u^\top,
\end{equation}
with $A_u$ and $G_u$ being the unstable parts of $A$ and $G$.

With the filtering result above, we now prove statement i) and show that $\bar{I}(E; X_0) \geq \sum_j \log |\lambda_j^{+}(A)| = (1/2)  \log \det \break  (G_u^{}   P_u^{-}G_u^\top+ I)$, where the inequality is proved by resorting to \cite[Lemma 4.1]{Martins_TAC_2007}, and the equality can be obtained by taking $\log\det(\cdot)$ on both sides of \eqref{LTI_ARE_Ctrl}. To prove statement ii) and show that $\bar{I}(E; X_0) \leq (1/2) \log\det(G_u^{}P_u^{-}G_u^\top + I) = \sum_j \log |\lambda_j^{+}(A)|$, consider the following manipulation on $I(E_0^n; X_0)$
	\begin{align}\label{prop36_eq7}
			I(E_0^n; X_0) & \neweq{(a)} h(E_0^n) - h(W_0^n) \allowdisplaybreaks \\
			& \hspace{-20pt} \neweq{(b)} h(G (X_0^n - \hat{X}_{0:n}^{-}) + W_0^n) - h(W_0^n) \nonumber \allowdisplaybreaks\\
			& \hspace{-20pt} \newleq{(c)} \frac{1}{2} \sum_{i=0}^{n} \log  [(2\pi {\rm e})^s  |GP_i^{-}G^\top + I|] - \frac{n+1}{2} \log( 2\pi {\rm e} )^s \nonumber \allowdisplaybreaks\\
			&\hspace{-20pt} \neweq{(d)}  (n+1) \sum_j \log|\lambda_j^{+}(A_u)| ,\nonumber
	\end{align}
	\noindent where (a) invokes \eqref{prop33_eq2} in \hyperref[prop33]{Proposition 3.3}; (b) follows from \eqref{LTI_Ctrl_Channel} and $h(GX_0^n + W_0^n) = h(G (X_0^n - \hat{X}_{0:n}^{-}) + W_0^n)$, since $\hat{X}_{0:n}^{-} =  \{ \mathbb{E}[X_{0}^i | E_0^{i-1}]  \}_{i=0}^n$ is deterministic; (c) follows from $h(X_0^n) \leq \sum_{i=0}^{n}h(X_i)$ and the maximum entropy condition when $E_i$ and $W_i \in \mathbb{R}^s$ in \eqref{LTI_Ctrl_Channel}, and (d) relies on \eqref{prop36_eq5} and \eqref{LTI_ARE_Ctrl}. Dividing both sides of \eqref{prop36_eq7} by $n+1$ and taking the limit superior as $n\rightarrow\infty$, we have $\bar{I}(E; X_0) \leq \sum_j \log|\lambda_j^{+}(A)|$. Combining statements i) and ii), we then prove the identity $\bar{I}(E; X_0) = \sum_j \log |\lambda_j^{+}(A)|$ in \hyperref[prop36]{Proposition 3.6}. More interesting observations on the lower and upper bounds in \eqref{prop36_eq1} can also be implied from \eqref{LTI_Prediction} and \eqref{LTI_Correction}.\footnote{\label{LTI_foot}For example, {\bf a)} $2\sum_j \log |\lambda_j^{+}(A)| = \sum_k \log|1+\eta_k|$ with $\{\eta_k\}$ being the eigenvalues of $G^{}_uP_u^{-}G_u^\top$, which is obtained by taking $\log \det(\cdot)$ on both sides of \eqref{LTI_ARE_Ctrl}; {\bf b)} $\sum_j \log |\lambda_j^{+}(A)| \geq \overline{\lim}_{n\rightarrow\infty}\sum_{i=0}^n {\rm cmmse}(U_i) / [2(n+1)]$, which is derived by pre- and post-multiplying both sides of \eqref{LTI_Correction} by $G$ and $G^\top$, taking the trace operator ${\rm tr}(\cdot)$, invoking inequality $x/(1+x) \leq \log(1+x)$ when $x > 0$ and result a), and taking average over $i = 0, \cdots, n$ as $n\rightarrow\infty$, and {\bf c)} $\sum_j \log |\lambda_j^{+}(A)| \leq \overline{\lim}_{n\rightarrow\infty} \sum_{i=0}^{n}{\rm pmmse}(U_i) / [2(n+1)]$, which is attained by taking $\log \det(\cdot)$ on both sides of \eqref{LTI_Prediction}, applying the relation $\log \det (I + G P^{-}_iG^\top) = -2 \log(  \mathbb{E}[ \exp[-(U_i - \hat{U}_i^{-})^\top (U_i - \hat{U}_i^{-}) / 2]  ]  ) \leq {\rm pmmse}(U_i)$ from \cite[Lemma 3.1]{Glover_SCL_1988} and Jensen's inequality, and taking average over $i = 0, \cdots, n$ as $n\rightarrow \infty$. Observations b) and c) along with the identity $\bar{I}(E; X_0) = \sum_j \log |\lambda_j^{+}(A)|$ reaffirm \eqref{prop36_eq1} in \hyperref[prop36]{Proposition 3.6}.}

	To verify the sandwich bounds in \eqref{prop36_eq1}, we consider a one-dimensional LTI control system with $A = 1.5$, $B = 1$, and $G = -1$ in \eqref{LTI_Ctrl_Plant}-\eqref{LTI_Ctrl_Channel}. By \hyperref[prop36]{Proposition 3.6}, the total information rate of channel \eqref{LTI_Ctrl_Channel} satisfies $\bar{I}(E; X_0) = 0.5850$. From \eqref{LTI_Prediction} and \eqref{LTI_Correction}, we figure out that the steady-state prediction MMSE is 1.25, and the steady-state causal MMSE is 0.5556, which together verify the inequalities in \eqref{prop36_eq1}.

	\section*{Appendix B: Stability and Sensitivity of Discrete-Time LTV Systems}\label{appB}
	Consider the following $m$-dimensional discrete-time homogeneous LTV system
	\begin{equation}\label{LTV_SYS}
		X_{i+1} = A_i X_i,
	\end{equation}
	where $X_i \in \mathbb{R}^{m}$, $\{A_i\}$ is a sequence of $m\times m$ matrices, and $i = k, k+1, \cdots, k+n$. Let $\varPhi(k+n, k) := \prod_{i=k}^{k+n-1}A_i$ be the discrete-time state transition matrix of \eqref{LTV_SYS} from $k$ to $k+n$. We say that the sequence $\{A_i\}$ is {\it uniformly exponentially stable} (UES), if there exist positive constants $\alpha$ and $\beta < 1$, independent of $k$ and $n$, such that $\| \varPhi_A(k + n, k) \| = \overline{\mu}(\varPhi_A(k + n, k)) < \alpha\beta^{n}$, and $\{A_i\}$ is {\it uniformly exponentially antistable} (UEA) if there exist positive constants $\alpha$ and $\beta > 1$, independent of $k$ and $n$, such that $\underline{\mu}(\varPhi_A(k + n, k)) > \alpha \beta^{n}$. 
	
	The LTV system \eqref{LTV_SYS} is said to possess an {\it exponential dichotomy} if it can be decomposed into stable and antistable parts. Let $\{P_k\}$ be a bounded sequence of projections in $\mathbb{R}^m$ such that the rank of $P_k$ is constant. We say that $\{P_k\}$ is a {\it dichotomy} for $\{A_i\}$ if the commutativity condition $A_kP_k = P_{k+1}A_k$ is satisfied for all $k$ and there exist positive constants $\alpha$ and $\beta > 1$ such that $\|\varPhi_A(k+n, k)P_k x\| > \alpha \beta^{n+1}  \|P_k x\|$ and $\|\varPhi_A(k+n, k)(I - P_k)x\| < (\alpha \beta^l)^{-1} \| (I - P_k) x\|$ for any $x$. 
	
	Define the weighted shift operator $\mathscr{S}_A$ by $(\mathscr{S}_A X)_i := A_{i-1} X_{i-1}$, and let $\sigma(\mathscr{S}_A)$ denote the spectrum of the operator. The sequence $\{A_k\}$ is stable (or antistable) if and only if $\sigma(\mathscr{S}_A) \in \mathbb{D}$ (or $\sigma(\mathscr{S}_A) \in \mathbb{C} \backslash \mathbb{D}$), where $\mathbb{D}$ denotes the unit disk, and $\mathbb{C}$ stands for the entire complex plane. The spectrum of the weighted shift is contained in $r$ concentric annuli:	
	\begin{equation*}
		\sigma(\mathscr{S}_A) = \cup_{j=1}^{r} \{ \lambda \in \mathbb{C}: \underline{\kappa}_j \leq |\lambda| \leq \overline{\kappa}_j \},
	\end{equation*}
	where $0 \leq \underline{\kappa}_1 \leq \overline{\kappa}_1 < \underline{\kappa}_2 < \cdots < \underline{\kappa}_r < \overline{\kappa}_r < \infty$ and $1 \leq r \leq m$. Exponential dichotomy is equivalent to the existence of an integer $q \in \{1, \cdots, r-1\}$ such that either $\overline{\kappa}_q < 1 < \underline{\kappa}_{q+1}$ or $\overline{\kappa}_r < 1$. More detailed definitions and explanations on the terminology related to the discrete-time LTV systems can be found in \cite{Iglesias_Auto_2001, Hu_LAA_2017} and references therein.

	\section*{Appendix~C: Proof of Corollary 3.12}\label{appC}
	The MMSE-based sandwich bounds in \eqref{cor312_eq1} can be readily derived by applying \hyperref[thm34]{Theorem 3.4} to the LTV plant \eqref{LTV_Ctrl_Plant} and control channel \eqref{LTI_Nonlinear_Channel}. It then remains to show that $\bar{I}(E; X_0) \geq \sum_{j=1}^{l}m_j \log \underline{\kappa}_j$, which, analogous to the constraint of LTV Bode's integral $\mathscr{B}$, indicates the control trade-off property of $\bar{I}(E; X_0)$, and is also an intermediate result invoked in the proof of \hyperref[prop311]{Proposition 3.11}. To prove this inequality, consider the antistable part of \eqref{LTV_Ctrl_Plant} governed by the following equation
	\begin{equation}\label{cor312_eq2}
		X_u(i+1) = A_u(i)X_u(i) + B_u(i)E(i),
	\end{equation}
	\noindent where $A_u(i)$ and $B_u(i)$ follow the definitions in \hyperref[lem39]{Lemma 3.9}. The solution to \eqref{cor312_eq2}, $X_u(n)$ or $X_u(n+1)$, is
	\begin{align}\label{cor312_eq3}
		X_u(n) & = \varPhi_{A_u}(n, 0) X_u(0) + \sum_{i=0}^{n-1}\varPhi_{A_u}(n, i+1)B_u(i)E(i) \nonumber \\
		& = \varPhi_{A_u}(n, 0) [X_u(0) + \bar{X}_u(n)],
	\end{align}
	\noindent where the state transition matrix is given by $\varPhi_{A_u}(n, 0) := \prod_{i=0}^{n-1}A_u(i)$, and $\bar{X}_u(n) := \varPhi_{A_u}(n, 0)^{-1} \sum_{i=0}^{n-1}\varPhi_{A_u}(n, i+1)B_u(i)E(i)$. Since the LTV subsystem \eqref{cor312_eq2}, stabilized by the nonlinear control mapping \eqref{Nonlinear_Ctrl} or \eqref{LTI_Nonlinear_Channel}, is internally mean-square stable, i.e., $\sup_i \mathbb{E}[X^\top_i X_i^{}] < \infty$, substituting \eqref{cor312_eq3} into the mean-square stability condition gives
	\begin{equation}\label{cor312_eq4}
		\begin{split}
			+\infty > S & > \log \det \mathbb{E}[X_u(n)X_u^\top(n)]\\
			& = 2 \log \det \varPhi_{A_u}(n, 0) + \log\det \Gamma_n,
		\end{split}
	\end{equation}
	\noindent where $\Gamma_n := \mathbb{E}[(X_u(0) + \bar{X}_u(n))(X_u(0) + \bar{X}_u(n))^\top]$. Moreover, the total information $I(E_0^{n-1}; X_0)$ in \eqref{LTI_Nonlinear_Channel} satisfies
	\begin{align}\label{cor312_eq5}
		I(E_0^{n-1}; X_0)    
		& \newgeq{(a)} I(\bar{X}_u(n); X_u(0))  \allowdisplaybreaks\\
		& = h(X_u(0)) - h(X_u(0) | \bar{X}_u(n))  \nonumber  \allowdisplaybreaks\\
		& \neweq{(b)} h(X_u(0)) - h(X_u(0) + \bar{X}_u(n) | \bar{X}_u(n)) \nonumber  \allowdisplaybreaks\\
		& \newgeq{(c)} h(X_u(0)) - h(X_u(0) + \bar{X}_u(n)) \nonumber  \allowdisplaybreaks\\
		& \newgeq{(d)} h(X_u(0)) - \frac{s_u}{2} \log(2\pi {\rm e}) - \frac{1}{2} \log \det \Gamma_n ,\nonumber
	\end{align}
	\noindent where (a) follows from the fact that $X_u(0)$ is a function of $X(0)$, and $\bar{X}_u(n)$ is a function of $E_0^{n-1}$; (b) uses the property $h(X|Y) = h(X + f(Y)|Y)$; (c) can be explained by the non-negativity of mutual information, since $I(X; Y) = h(X)-h(X|Y) \geq 0$, and (d) assumes that $X_u \in \mathbb{R}^{s_u}$ and resorts to the maximum condition of differential entropy. Substituting inequality \eqref{cor312_eq4} into \eqref{cor312_eq5} gives
	\begin{align}\label{cor312_eq6}
		I(E_0^{n-1}; X_0) \geq h(X_u(0)) & - ({s_u}/{2})  \log(2\pi e) \allowdisplaybreaks\\
		& - {S}/{2} + \log\det \varPhi_{A_u}(n, 0). \nonumber
	\end{align}
	Dividing both sides of \eqref{cor312_eq6} by $n$ and taking the limit superior as $n\rightarrow\infty$, we have
	\begin{align*}
		\bar{I}(E; X_0) \geq \uplim_{n\rightarrow\infty} \frac{\log \det \varPhi_{A_u}(n, 0)}{n} &  \geq \sum_{j=1}^{l} m_j \ln \underline{\kappa}_j \geq m_u \ln \beta.
	\end{align*}
	This completes the proof.

\section*{Appendix D: Proof of Theorem~4.4}\label{appD}
	We first derive the lower and upper bounds of $\bar{I}(Y; W, X_0)$ in \eqref{thm44_eq1}. From \eqref{thm42_eq2}, we can tell that $\uplim_{n\rightarrow\infty} [2(n+1)]^{-1} \cdot \sum_{i=0}^{n}{\rm cmmse}(Z_i)$ and $\uplim_{n\rightarrow\infty} [2(n+1)]^{-1}  \sum_{i=0}^{n}{\rm pmmse}(Z_i)$ are respectively a lower bound and an upper bound for $\bar{I}(Y; W, X_0)$. By \eqref{prop43_eq1} and the non-negativity of mutual information, or $I(Y_0^n; W_0^{n-1}  |  X_0) \geq 0$, we see that $\bar{I}(Y; X_0)$ is also a lower bound of $\bar{I}(Y; W, X_0)$. Meanwhile, according to the definition of filtering capacity, $\mathcal{C}_f$ should be the lowest upper bound for $\bar{I}(Y; W, X_0)$.

	We then consider the scenario when the noise-free output $Z_i$ is subject to the power constraint, i.e., $\mathbb{E}[Z_i^\top  Z_i] < \rho_i$, $\forall i$. Without loss of generality, suppose $Y_i$ and $V_i \in \mathbb{R}^s$ in \eqref{measure_channel}. To verify $\mathcal{C}_f = \mathcal{C} = \overline{\lim}_{n\rightarrow\infty}\sup_{\mathbb{E}[Z_i^\top Z_i] \leq \rho_i}  [2(n+1)]^{-1} \break \sum_{i=0}^{n} \log(|\Sigma_{Z_i}  + I|)$ in \hyperref[thm44]{Theorem~4.4}, we first show that $\mathcal{C} \leq \mathcal{C}_f = \overline{\lim}_{n\rightarrow\infty}  \sup_{\mathbb{E}[Z_i^\top Z_i] \leq \rho_i} [2(n+1)]^{-1} \sum_{i=0}^{n} \log(|\Sigma_{Z_i} + I|)$, where the inequality can be implied by the definition of $\mathcal{C}$ and $\mathcal{C}_f$ and the identity between total information and input-output mutual information in \eqref{measure_channel}, i.e., ${I}(Y_0^n; W_0^{n-1}, X_0) = I(Y_0^n; Z_0^n)$, and the equality follows 
	\begin{align} \label{thm44_eq3}
		\bar{I}(Y; W, X_0) & \neweq{(a)} \bar{h}(Y) - \bar{h}(V) \leq \mathcal{C}_f \allowdisplaybreaks \\
		& \hspace{-40pt} \newleq{(b)} \uplim_{n\rightarrow\infty} \sup_{\mathbb{E}[Z_i^\top Z^{}_i] \leq \rho^{}_i} \frac{\sum_{i=0}^{n}\log[(2\pi {\rm e})^{s} | {\Sigma_{Y_i}}|] }{2(n+1)} - \frac{\log [(2\pi {\rm e})^{s}]}{2} \nonumber \allowdisplaybreaks \\
		&\newleq{(c)} \uplim_{n\rightarrow\infty} \sup_{\mathbb{E}[Z_i^\top Z^{}_i] \leq \rho^{}_i} \frac{1}{2(n+1)}  \sum_{i=0}^{n} \log(|\Sigma_{Z_i} \break +I|) \nonumber \allowdisplaybreaks \\
		&\neweq{(d)} \uplim_{n\rightarrow \infty} \frac{1}{2(n+1)} \sum_{i=0}^{n} \log \left( 1 + \rho^{}_i \right) , \nonumber
	\end{align}
	\noindent where (a) follows from \eqref{prop43_eq2} and \hyperref[def21]{Definitions~2.1} and \ref{def22}; (b) uses the property $h(Y_0^n) \leq \sum_{i=0}^{n}h(Y_i)$ with equality attained if and only if $\{Y_i\}_{i=0}^n$ are mutually independent, maximum entropy condition, and $V_i \sim \mathcal{N}(0, I)$; (c) utilizes the fact that $\Sigma_{Y_i} = \Sigma_{Z_i} + \Sigma_{V_i}$, where $\Sigma_{V_i} = I$, and (d) holds when \eqref{measure_channel} is one-dimensional, i.e., $s = 1$. To show the other direction, consider a special scenario when the noise-free outputs $Z_i \sim \mathcal{N}(0, \Sigma_{Z_i})$ for $i = 0, \cdots, n$ in the non-feedback channel \eqref{measure_channel} are mutually independent. The channel capacity $\mathcal{C}$ of this special scenario satisfies
	\begin{equation}\label{thm44_eq4}
		\mathcal{C} \geq \bar{I}(Y; Z) = \uplim_{n\rightarrow\infty} \sup_{\mathbb{E}[Z_i^\top Z^{}_i] \leq \rho^{}_i} \frac{1}{2(n+1)}  \sum_{i=0}^{n} \log(|\Sigma_{Z_i}  +I|).
	\end{equation}	
	Combining \eqref{thm44_eq3} and \eqref{thm44_eq4}, we prove the equations subject to the finite-power constraint in \hyperref[thm44]{Theorem~4.4}.

	Lastly, when the process noise $W_i \sim \mathcal{N}(0, \varepsilon^2 I)$ vanishes, i.e., $\varepsilon\rightarrow 0$ in \eqref{Filter_Plant}, the total information rate $\bar{I}(Y; W, X_0)$ in \eqref{thm42_eq2} satisfies
	\begin{equation}\label{thm44_eq5}
		\bar{I}(Y; W, X_0) \geq \lim_{\varepsilon \rightarrow 0} \bar{I}(Y; W, X_0) = \bar{I}(Y; X_0),
	\end{equation}
	\noindent which follows the decomposition in \eqref{prop43_eq1} and the fact that when $W_0^n \rightarrow 0$ and becomes deterministic, $\lim_{\varepsilon \rightarrow 0} I(Y_0^n; W_0^n|X_0) \break = 0$. Taking the limit $\varepsilon\rightarrow 0$ in \eqref{thm42_eq2} and substituting \eqref{thm44_eq5} into the result, we obtain \eqref{thm44_eq2}.

	\section*{Appendix E: Proof of Proposition 4.5}\label{appE}
	Since \hyperref[prop45]{Proposition~4.5} can be virtually proved by replacing the feedback control channel in the proof of \hyperref[prop36]{Proposition~3.6} or \hyperref[prop311]{Proposition~3.11} with the non-feedback measurement channel in \eqref{LTI_Filtering}, in this proof we only show the key results and skip the overlapped details. To calculate the optimal output estimates, $\hat{Z}_i^{-} := \mathbb{E}[Z_i|Y_0^{i-1}] = \mathbb{E}[HX_i|Y_0^{i-1}] $ and $\hat{Z}_i := \mathbb{E}[Z_i|Y_0^{i}]= \mathbb{E}[HX_i|Y_0^{i}]$, and the predicted/causal MMSEs of $Z_i$ in \eqref{prop45_eq1}, we can apply the LTI Kalman filter to \eqref{LTI_Filtering}. Define the state prior estimate $\hat{X}_i^{-} := \mathbb{E}[X_i | Y_0^{i-1}]$ and state posterior estimate $\hat{X}_{i} := \mathbb{E}[X_{i} | Y_0^{i}]$. We can calculate $\hat{Z}_i^{-}$, $\hat{X}_i^{-}$, and ${\rm pmmse}(Z_i)$ from the prediction step of Kalman filter
	\begin{equation}\label{LTI_fl_prediction}
		\begin{split}
			\hat{X}_{i}^{-} & = A \hat{X}_{i-1}\\
			P_{i}^{-} & = AP_{i-1}^{-}(HP_{i-1}^{-}H^\top + I)^{-1}A^\top + \varepsilon^2 BB^\top,
		\end{split}
	\end{equation}
	\noindent where the state prior error covariance is $P_i^{-} := \mathbb{E}[(X_i - \hat{X}_i^{-}) \cdot (X_i - \hat{X}_i^{-})^\top]$. $\hat{Z}_i$, $\hat{X}_i$. In the meantime, ${\rm cmmse}(Z_i)$ can be calculated from the correction step of Kalman filter
	\begin{equation}\label{LTI_fl_correction}
		\begin{split}
			\hat{X}_{i} & = \hat{X}_{i}^{-} + K_{i}(Y_{i} - H \hat{X}_{i}^{-})\\
			P_{i} & = P_{i}^{-} - P_{i}^{-}H^\top  (HP_{i}^{-}H^\top + I)^{-1} H P_{i}^{-},
		\end{split}
	\end{equation}
	\noindent where the Kalman gain is given by $K_{i} := P_{i}^{-}H^\top  (I + HP_{i}^{-}H^\top)^{-1}$, and the state posterior error covariance is $P_{i} := \mathbb{E}[(X_{i} - \hat{X}_{i}) (X_{i} - \hat{X}_{i})^\top]$.

	Inequalities \eqref{prop45_eq1} are obtained by applying \hyperref[thm44]{Theorem~4.4} to the measurement channel in \eqref{LTI_Filtering}. We then prove $\bar{I}(Y; X_0) = \sum_j \log |\lambda_j^{+}(A)|$ by respectively showing that  i) $\bar{I}(Y; X_0) \geq \sum_j \log |\lambda_j^{+}(A)|$, and ii) $\bar{I}(Y; X_0) \leq \sum_j \log |\lambda_j^{+}(A)|$. Statement i) can be obtained by using \hyperref[ass4]{(A4)} and mimicking the proof of \hyperref[cor312]{Corollary~3.12}. Meanwhile, in the limits of $i\rightarrow \infty$ and $\varepsilon \rightarrow 0$, by taking $\log \det( \cdot )$ on both sides of \eqref{LTI_fl_prediction} and applying the LTI version of \hyperref[lem46]{Lemma~4.6} to the prior error covariance matrices in \eqref{LTI_fl_prediction}, we have
	\begin{align}\label{prop45_eq2}
			\lim_{\varepsilon \rightarrow 0} \lim_{i\rightarrow\infty}\log \det(HP_{i}^{-}H^{\top} + I)  & = \log \det(H^{}_uP_u^{-}H_u^{\top} + I) \nonumber \allowdisplaybreaks \\
			& = 2\sum_j \log |\lambda_j^{+}(A)| ,
	\end{align}
	\noindent where $P_u^{-}$ and $H_u$ are the unstable parts of $\lim_{i\rightarrow\infty}P_i^{-}$ and $H$, and $\lambda_j^{+}(A)$ denote the eigenvalues of $A_u$, which is the unstable part of $A$. Supposing $Y_i$ and $V_i \in \mathbb{R}^s$ in \eqref{LTI_Filtering}, and following \eqref{prop43_eq2}, \eqref{prop45_eq2} and the similar arguments for \eqref{prop36_eq7}, we have
	\begin{align}\label{prop45_eq3}
		I(Y_0^n; X_0) & = h(Y_0^n) - h(V_0^n) \allowdisplaybreaks\\
		& = h(H(X_0^n - \hat{X}_{0:n}^{-}) + V_0^n) - h(V_0^n) \nonumber \allowdisplaybreaks\\
		& \leq \frac{1}{2} \sum_{i=0}^{n} \log[(2\pi {\rm e})^s |HP_i^{-}H^\top + I|] - \frac{n}{2} \log (2\pi {\rm e})^s \nonumber \allowdisplaybreaks\\
		&\overset{\varepsilon \rightarrow 0 \atop n\rightarrow \infty}{=} n \sum_j \log |\lambda_j^{+}(A)|.  \nonumber
	\end{align}
	\noindent Statement ii) and thus \hyperref[prop45]{Proposition~4.5} can be proved by dividing both sides of \eqref{prop45_eq3} by $n$ and taking the limit superior (or limit for the stationary case) as $n\rightarrow \infty$. This completes the proof.

	\section*{Appendix F: Proof of Lemma 4.6}\label{appF}
	For the discrete-time filtering system \eqref{LTV_Filtering} satisfying \hyperref[ass2]{(A2)} and \hyperref[ass3]{(A3)}, we can decompose it into stable and antistable parts by \hyperref[lem39]{Lemma 3.9}:
	\begin{align}\label{lem46_eq3}
			\hspace{-2pt}\left[\begin{matrix}
				X_s(i'+1)\\
				X_u(i'+1)					
			\end{matrix}\right] &= \left[\begin{matrix}
				A_s(i') & 0\\
				0 & A_u(i')
			\end{matrix}\right]\left[ \begin{matrix}
				X_s(i')\\
				X_u(i')
			\end{matrix}\right] + \left[\begin{matrix}
				B_s(i')\\
				B_u(i')
			\end{matrix}\right] W(i'), \nonumber \\
			Y(i') & = \left[\begin{matrix}
				H_s(i') & H_u(i')
			\end{matrix}\right] \left[  \begin{matrix}
				X_s(i')\\
				X_u(i')
			\end{matrix}\right] + V(i'),
	\end{align}
	\noindent where $X_s(i')$ and $X_u(i')$ respectively denote the stable and antistable modes. Since the original system \eqref{LTV_Filtering} is uniformly completely stabilizable and reconstructible, the transformed system and its stable/antistable subsystems in \eqref{lem46_eq3} are also stabilizable and reconstructible. Hence, in the limit of $i'\rightarrow \infty$, the asymptotic state  prior error covariance $\bar{P}^{-}_i := \lim_{i' \rightarrow\infty}  \mathbb{E}\{[\tilde{X}_s^{-}(i')^\top, \tilde{X}_u^{-}(i')^\top]^\top  [\tilde{X}_s^{-}(i')^\top,   \tilde{X}_u^{-}(i')^\top]\}$ of filtering problem \eqref{lem46_eq3} satisfies \cite[Theorem 6.45]{Sivan_1972}
	\begin{equation}\label{lem46_eq4}
		\bar{P}^{-}_{i+1} = \bar{A}^{}_{i} \bar{P}_{i}^{-}( \bar{H}_{i}^{}\bar{P}_{i}^{-}\bar{H}_{i}^\top + I)^{-1}\bar{A}^\top_{i} + \varepsilon^2 \bar{B}^{}_{i} \bar{B}_{i}^\top,
	\end{equation}
	\noindent where $\bar{A}_i = {\rm diag}\{ A_s(i), A_u(i) \}$, $\bar{B}_{i} = [B_s^\top(i), B_u^\top(i)]^\top$, and $\bar{H}_{i} = [H_s(i), H_u(i)]$. When the process noise $W_i$ vanishes in \eqref{LTV_Filtering} and \eqref{lem46_eq3}, or in the limit of $\varepsilon \rightarrow 0$, \eqref{lem46_eq4} becomes a singularly perturbed equation, and we seek its solution $\bar{P}_i^{-}$ in the form of power series in $\varepsilon$ \cite{Naidu_1988, Braslavsky_Auto_1999}:
	\begin{equation}\label{lem46_eq5}
		\bar{P}_i^{-}= \left[\begin{matrix}
			\varepsilon^2 \bar{P}^{-}_s(i) + O(\varepsilon^3) & \varepsilon^2 \bar{P}^{-}_0(i) + O(\varepsilon^3)\\
			* & \bar{P}^{-}_u(i) + O(\varepsilon)
		\end{matrix}\right].
	\end{equation}
	\noindent By letting 
	\begin{equation}\label{lem46_eq6}
		(\bar{H}_{i}^{}\bar{P}_{i}^{-} \bar{H}_{i}^\top + I)^{-1} 
		=:
		\left[\begin{matrix}
			F_s(i) & F_0(i)\\
			F_0(i) & F_u(i)
		\end{matrix}\right] ,
	\end{equation}
	\noindent and substituting \eqref{lem46_eq5} and \eqref{lem46_eq6} into \eqref{lem46_eq4}, we have
	\begin{subequations}\label{lem46_eq7}
		\begin{align}
			&\hspace{-5pt}\varepsilon^2 \bar{P}_s^{-}(i+1) = \varepsilon^2 [ A^{}_s(i) \bar{P}_s^{-}(i) F^{}_s(i) A_s^\top(i) + A^{}_s(i) \bar{P}_0^{-}(i) \cdot \nonumber \allowdisplaybreaks \\
			&\hspace{70pt} F^{}_0(i) A_s^\top (i) + B^{}_s(i) B_s^\top(i) ] + O(\varepsilon^3) \label{lem46_eq7a} \allowdisplaybreaks\\
			&\hspace{-5pt}\varepsilon^2 \bar{P}^{}_0(i+1) = \varepsilon^2 [ A^{}_s(i)\bar{P}_s^{-}(i) F^{}_0(i) A_u^\top(i) + A^{}_s(i) P_0^{-}(i) \cdot \nonumber \allowdisplaybreaks \\
			&\hspace{70pt} F^{}_u(i) A_u^\top(i) + B^{}_s(i) B_u^\top(i) ] + O(\varepsilon^3) \label{lem46_eq7b} \allowdisplaybreaks\\
			&\hspace{-5pt}\bar{P}_u^{-}(i+1) = \varepsilon^2 [A^{}_u(i) \bar{P}_0^{-}(i) F^{}_0(i) A_u^\top(i) + B^{}_u(i) B_u^\top(i)] \nonumber  \allowdisplaybreaks \\
			&\hspace{70pt} + A^{}_u(i) \bar{P}_u^{-}(i) F^{}_u(i) A_u^\top(i) + O(\varepsilon^3). \label{lem46_eq7c}
		\end{align}
	\end{subequations}
	\noindent After setting $\varepsilon = 0$ in \eqref{lem46_eq5} and \eqref{lem46_eq7}, the only non-zero element in $\bar{P}_i^{-}$ is $\bar{P}^{-}_u(i)$. Since $F^{}_u(i) = [H_u(i) \bar{P}^{-}_u(i) H^\top_u(i) + I]^{-1}$ in \eqref{lem46_eq7c}, we can then imply the asymptotic prior error covariance in \eqref{lem46_eq1}. Substituting the asymptotic prior error covariance \eqref{lem46_eq1} into the correction step of LTV Kalman filter \cite[Sec. 2.7.2]{Lewis_2017}, $P_i = P_i^{-} - P_i^{-}H_i^\top(H^{}_iP_i^{-}H_i^\top + I)^{-1}H^{}_i P_i^{-}$, we obtain the asymptotic posterior error covariance in \eqref{lem46_eq2}.

	\section*{Appendix G: Proof of Proposition 4.7}\label{appG}
Since \hyperref[prop47]{Proposition 4.7} and this proof can be seen as the filtering and non-feedback counterparts of \hyperref[prop311]{Proposition 3.11} and its proof, in the following, we only show the key steps while skip the overlapped details. As the process and measurement noises in \eqref{LTV_Filtering} are uncorrelated and zero-mean white Gaussian, we can use the LTV Kalman filter to calculate the optimal estimates of $Z_i$, $\hat{Z}_i^{-} := \mathbb{E}[Z_i|Y_0^{i-1}] = \mathbb{E}[H_iX_i|Y_0^{i-1}]$ and $\hat{Z}_i :=   \mathbb{E}[Z_i|  Y_0^i] = \mathbb{E}[H_iX_i|Y_0^i]$, and the predicted/causal MMSEs of $Z_i$ in \eqref{prop47_eq1}. Let the state prior estimate be $\hat{X}_i^{-} := \mathbb{E}[X_i | Y_0^{i-1}]$ and the state posterior estimate be $\hat{X}_i := \mathbb{E}[X_i | Y_0^i]$. We can then calculate $\hat{Z}_i^{-}$, $\hat{X}_i^{-}$, ${\rm pmmse}(Z_i)$ in \eqref{prop47_eq1}, and the state prior error covariance $P_i^{-} := \mathbb{E}[(X_i - \hat{X}^{-}_i)(X_i - \hat{X}^{-}_i)^\top]$ from the prediction step of Kalman filter
		\begin{equation}\label{LTV_Filter_Prediction}
			\begin{split}
				\hat{X}_{i+1}^{-} & = A_{i} \hat{X}_{i},\\
			P_{i+1}^{-} & = A^{}_{i} P_{i}^{-} (H^{}_{i} P_{i}^{-} H_{i}^\top + I)^{-1}A_{i}^\top + \varepsilon^2 B^{}_{i} B_{i}^\top.
			\end{split}
		\end{equation}
		Meanwhile, $\hat{Z}_i$, $\hat{X}_i$, ${\rm cmmse}(Z_i)$ in \eqref{prop47_eq1}, and state posterior error covariance $P_i := \mathbb{E}[(X_i - \hat{X}_i)(X_i - \hat{X}_i)^\top]$ can be calculated from the correction step of Kalman filter
		\begin{equation}\label{LTV_Filter_Correction}
			\begin{split}
				\hat{X}^{}_{i} & = \hat{X}_{i}^{-} + K_{i}^{}(Y^{}_{i} - H^{}_{i} \hat{X}_{i}^{-}) ,\\
				P^{}_{i} & = P_{i}^{-} - P_{i}^{-}H_{i}^\top  (H_{i}^{} P_{i}^{-}H_{i}^\top + I)^{-1} H^{}_{i} P_{i}^{-},
			\end{split}
		\end{equation}
		\noindent where $K_{i} := P_{i}^{-}H^\top_{i}(H^{}_{i}P_{i}^{-}H_{i}^\top + I)^{-1}$ is the Kalman gain.

		With above estimates, we now prove the results in \hyperref[prop47]{Proposition 4.7}. The MMSE-based sandwich bounds in \eqref{prop47_eq1} can be derived by applying \hyperref[thm44]{Theorem 4.4} to the measurement channel \eqref{LTV_Filtering}. Meanwhile, once we show that $\bar{I}(Y; X_0) = \overline{\lim}_{n\rightarrow\infty}(n+ \break 1)^{-1} \log\det[\varPhi_{A_u}(n+1, 0)]$, we can tell that $\sum_{j=1}^{l} m_j  \log \underline{\kappa}_j$ is also a lower bound of $\bar{I}(Y; X_0)$ by \hyperref[lem310]{Lemma 3.10}. Hence, to prove \hyperref[prop47]{Proposition~4.7}, it remains to verify that $\bar{I}(Y; X_0) = \overline{\lim}_{n\rightarrow\infty}(n+1)^{-1} \log\det[\varPhi_{A_u}(n+1, 0)]$, which is equivalent to showing that i) $\bar{I}(Y; X_0) \geq \overline{\lim}_{n\rightarrow\infty}(n+1)^{-1} \cdot \break \log\det[\varPhi_{A_u}( n+1,  0)]$, and ii) $\bar{I}(Y; X_0) \leq \overline{\lim}_{n\rightarrow\infty}(n+1)^{-1}\cdot\break \log\det[\varPhi_{A_u}(  n+1,  0)]$. Statement i) can be proved by using the\break finite-power constraint \hyperref[ass4]{(A4)} and applying (the proof of) \hyperref[cor312]{Corollary 3.12} to \eqref{LTV_Filtering}. Taking $\log \det(\cdot)$ on both sides of \eqref{LTV_Filter_Prediction}, summing up the result from $i = 0$ to $n$, and applying \hyperref[lem46]{Lemma 4.6} to $P_i^{-}$, as $\varepsilon \rightarrow 0$ and $n\rightarrow \infty$, we have
		\begin{align}\label{prop47_eq2}
			\uplim_{n\rightarrow\infty}  \sum_{i=0}^{n} \log\det \varPsi(i) &= 2 \uplim_{n\rightarrow\infty} \log\det\varPhi_{A_u}(n+1, 0) \allowdisplaybreaks \\
			& \hspace{7pt}+  \log\det \bar{P}_{u}^{-}(0) -  \log\det \bar{P}_{u}^{-}(n+1), \nonumber
		\end{align}
		\noindent where $\varPsi(i)  := H^{}_u(i)\bar{P}^{-}_u(i)H_u^\top(i) + I$; the asymptotic prior error covariance of antistable mode $\bar{P}^{-}_u(i)$ is defined in \eqref{lem46_eq1}, and $H_u(i)$ denotes the antistable part of $H_i$ in \eqref{LTV_Filtering}. Suppose $Y_i$ and $V_i \in \mathbb{R}^s$ in \eqref{LTV_Filtering}. By using \eqref{prop43_eq2} and \eqref{prop47_eq2}, and following the similar derivations of \eqref{prop311_eq8} and \eqref{prop45_eq3}, we have
		\begin{align}\label{prop47_eq3}
			I(Y_0^n; X_0) & = h(Y_0^n) - h(V_0^n) \allowdisplaybreaks\\
			&\hspace{-5pt} = h(H(X_0^n - \hat{X}_{0:n}^{-}) + V_0^n) - h(V_0^n) \nonumber \allowdisplaybreaks\\
			&\hspace{-5pt} \leq \frac{1}{2} \sum_{i=0}^{n} \log[(2\pi {\rm e})^s|H^{}_i P_i^{-} H_i^\top + I|] - \frac{n}{2} \log (2\pi {\rm e})^s \nonumber \allowdisplaybreaks\\
			& \hspace{-5pt} \overset{\varepsilon \rightarrow 0 \atop n\rightarrow \infty}{=}  \log\det\varPhi_{A_u}(n+1, 0) + \frac{1}{2} \log \det \bar{P}_u^{-}(0) \nonumber \allowdisplaybreaks \\
			& \hspace{105pt} - \frac{1}{2} \log\det \bar{P}_u^{-}(n+1) ,  \nonumber
		\end{align}
		where the last equality follows \hyperref[lem46]{Lemma 4.6} and \eqref{prop47_eq2}. Statement ii) and hence \hyperref[prop47]{Proposition 4.7} are proved by dividing \eqref{prop47_eq3} by $n+1$ and taking the limit superior (or limit for the stationary case) as $n\rightarrow \infty$. This completes the proof.

\end{document}